\documentclass[conference]{IEEEtran}
\usepackage{graphicx}
\usepackage{amsthm}
\usepackage{epsfig}
\usepackage{latexsym}
\usepackage{amsfonts}
\usepackage{here}
\usepackage{rawfonts}
\usepackage[latin1]{inputenc}
\usepackage[T1]{fontenc}
\usepackage{calc}
\usepackage{capitalgreekitalic}
\usepackage{url}
\usepackage{enumerate}
\usepackage{color}
\usepackage[tbtags]{amsmath}
\usepackage{amssymb}
\usepackage{upref}
\usepackage{epic,eepic}
\usepackage{times}
\usepackage{dsfont}
\usepackage{comment}
\usepackage{cite}

\renewcommand{\Pr}{\ensuremath{\operatorname{Pr}}}

\newtheorem{theorem}{\bf Theorem}
\newtheorem{proposition}{\bf Proposition}

\usepackage{dsfont}

\newcounter{step}
\newlength{\totlinewidth}
\newenvironment{algorithm}{%
  \rule{\linewidth}{1pt}
  \begin{list}{}%
    {\usecounter{step}%
      \settowidth{\labelwidth}{\textbf{Step 2:}}%
      \setlength{\leftmargin}{\labelwidth}%
      \setlength{\topsep}{-2pt}%
      \addtolength{\leftmargin}{\labelsep}%
      \addtolength{\leftmargin}{2mm}%
      \setlength{\rightmargin}{2mm}%
      \setlength{\totlinewidth}{\linewidth}%
      \addtolength{\totlinewidth}{\leftmargin}%
      \addtolength{\totlinewidth}{\rightmargin}%
      \setlength{\parsep}{0mm}%
      \raggedright}}%
  {\end{list}%
  \rule{\linewidth}{1pt}}
\newcounter{substep}

  {\end{list}}

\newlength{\aligntop}
\newlength{\alignbot}
 \makeatletter
\renewenvironment{align}{%
  \vspace{\aligntop}
  \start@align\@ne\st@rredfalse\m@ne
}{%
  \math@cr \black@\totwidth@
  \egroup
  \ifingather@
    \restorealignstate@
    \egroup
    \nonumber
    \ifnum0=`{\fi\iffalse}\fi
  \else
    $$%
  \fi
  \ignorespacesafterend%
  \vspace{\alignbot}\par\noindent
} \makeatother

\IEEEoverridecommandlockouts

\usepackage{algorithm}%,algpseudocode}
\usepackage{algorithmic}
\usepackage{subfigure}
\begin{document}
%\pagenumbering{gobble}% Remove page numbers (and reset to 1)
\clearpage
%\maketitle
%Title.
% ------
\title{\huge Echo State Networks for Proactive Caching in Cloud-Based Radio Access Networks with Mobile Users}
%
% Single address.
% ---------------
\author{\IEEEauthorblockN{\large Mingzhe Chen$^1$, Walid Saad$^{2}$, Changchuan Yin$^1$, and M\'erouane Debbah$^{3,4}$}\vspace{-0.3cm}\\
\IEEEauthorblockA{
\small $^1$ Beijing Laboratory of Advanced Information Network, Beijing University of Posts and Telecommunications, Beijing, China 100876,\\ Emails: \url{chenmingzhe@bupt.edu.cn}, \url{ccyin@ieee.org.}\\
		$^2$ Wireless@VT, Electrical and Computer Engineering Department, Virginia Tech, VA, USA, Emails:\url{walids@vt.edu}.\\
		$^3$ Large Networks and Systems Group (LANEAS), CentraleSup\'elec, Universit\'e Paris-Saclay, Gif-sur-Yvette, France.\\
		$^4$ Mathematical and Algorithmic Sciences Lab, Huawei France R \& D, Paris, France, Email: \url{merouane.debbah@huawei.com}.\\
		%\thanks{This paper will be presented in part at the IEEE GLOBECOM conference, Washington, DC, USA, 2016 \cite{MozaffariIoT}.}
		%\thanks{This work was supported by the U.S. National Science
		%Foundation under Grants AST-1506297, by the Office of Naval Research (ONR) under Grant N00014-15-1-2709, and, by the ERC Starting
		%Grant 305123 MORE (Advanced Mathematical Tools for Complex Network
     	%Engineering), and by the Academy of Finland.}
}\vspace{-0.6cm}}
%
% For example:
% ------------
%\address{School\\
%   Department\\
%   Address}
%
% Two addresses (uncomment and modify for two-address case).
% ----------------------------------------------------------
%\twoauthors
%  {A. Author-one, B. Author-two\sthanks{Thanks to XYZ agency for funding.}}
%   {School A-B\\
%   Department A-B\\
%   Address A-B}
%  {C. Author-three, D. Author-four\sthanks{The fourth author performed the work
%   while at ...}}
%   {School C-D\\
%   Department C-D\\
%   Address C-D}
%
%\tableofcontents
%\pdfbookmarks
%\ninept
%

\maketitle

\vspace{0cm}
\begin{abstract}
%Proactive caching at the base band units (BBUs) in cloud radio access networks (CRANs) has attracted significant attention. However, most existing works assume certain content distribution while ignoring the massive nature of data in CRANs. Moreover, the user mobility has rarely been considered by recent works which contains useful information that can be developed to improve the caching efficiency.  
%In contrast, 
In this paper, the problem of proactive caching is studied for cloud radio access networks (CRANs). In the studied model, the baseband units (BBUs) can predict the content request distribution and mobility pattern of each user, determine which content to cache at remote radio heads and BBUs. This problem is formulated as an optimization problem which jointly incorporates backhaul and fronthaul loads and content caching.
To solve this problem, an algorithm that combines the machine learning framework of \emph{echo state networks} with sublinear algorithms is proposed. Using  echo state networks (ESNs), the BBUs can predict each user's content request distribution and mobility pattern while having only limited information on the network's and user's state. In order to predict each user's periodic mobility pattern with minimal complexity, the memory capacity of the corresponding ESN is derived for a periodic input. This memory capacity is shown to capture the maximum amount of user information needed for the proposed ESN model. Then, a sublinear algorithm is proposed to determine which content to cache while using limited content request distribution samples. Simulation results using real data from \emph{Youku} and the \emph{Beijing University of Posts and Telecommunications} show that the proposed approach yields significant gains, in terms of sum effective capacity, that reach up to $27.8\%$ and $30.7\%$, respectively, compared to random caching with clustering and random caching without clustering algorithm. 
\end{abstract}
\vspace{0cm}
\renewcommand{\thefootnote}{\fnsymbol{footnote}}
\footnotetext{A preliminary version of this work \cite{EchoStateNetworks} was submitted to IEEE GLOBECOM Workshops.\\}
\vspace{0cm}
\renewcommand{\thefootnote}{\fnsymbol{footnote}}
\footnotetext[1]{ This work was supported in part by the National Natural Science Foundation of China under Grants 61671086, 61629101, by the ERC Starting Grant 305123 MORE (Advanced Mathematical Tools for Complex Network Engineering) and by the U.S. National Science Foundation under Grants IIS-1633363, CNS-1460316 and CNS-1513697.}

\vspace{-0cm}
\section{Introduction}
Cellular systems based on cloud radio access networks (CRANs) enable communications using a massive number of remote radio heads (RRHs) are controlled by cloud-based baseband units (BBUs) via wired or wireless fronthaul links \cite{2}. These RRHs act as distributed antennas that can service the various wireless users. 
To improve spectral efficiency, cloud-based cooperative signal processing techniques can be executed centrally at the BBUs \cite{MugenRecent}. 
However, despite the ability of CRAN systems to run such complex signal processing functions centrally, their performance remains limited by the capacity of the fronthaul and backhaul (CRAN to core) links \cite{MugenRecent}. Indeed, given the massive nature of a CRAN, relying on fiber fronthaul and backhaul links may be infeasible. Consequently, capacity-limited wireless or third party wired solutions for the backhaul and fronthaul connections are being studied for CRANs such as in \cite{Jointwireless} and \cite{Optimalfronthaul}.
To overcome these limitations, one can make use of content caching techniques\cite{Cluster,Cooperative,MeanField,Content,BackhaulChen} in which users can obtain contents from storage units deployed at cloud or RRH level.
However, deploying caching strategies in a CRAN environment faces many challenges that include optimized cache placement, cache update, and accurate prediction of content popularity.    
 %Such techniques have become recently very popular, particularly with the emergence of social networks.
 
The existing literature has studied a number of problems related to caching in CRANs, {heterogeneous networks, and content delivery networks (CDNs) \cite{Cluster,Cooperative,MeanField,Content,BackhaulChen,Tran2016Octopus,Cachingimprovement,Jointcaching,ContextAware,Sung2016Efficient,kang2014mobile,De2011Optimum}}. In \cite{Cluster}, the authors study the effective capacity of caching using stochastic geometry and shed light on the main benefits of caching. The work in \cite{Cooperative} proposes a novel cooperative hierarchical caching framework for the CRAN to improve the hit ratio of caching and reduce backhaul traffic load by jointly caching content at both the BBU level and RRH level. In \cite{MeanField}, the authors analyzed the asymptotic limits of caching using mean-field theory. The work in \cite{Content} introduces a novel approach for dynamic content-centric base station clustering and multicast beamforming that accounts for both channel condition and caching status. In \cite{BackhaulChen}, the authors study the joint design of multicast beamforming and dynamic clustering to minimize the power consumed, while quality-of-service (QoS) of each user is guaranteed and the backhaul traffic is balanced. {The authors in \cite{Tran2016Octopus} propose a novel caching framework that seeks to realize the potential of CRANs by using a cooperative hierarchical caching approach that minimizes the content delivery costs and improves the users quality-of-experience.} In \cite{Cachingimprovement}, the authors develop a new user clustering and caching method according to the content popularity. The authors also present a method to estimate the number of clusters within the network based on the Akaike information criterion. In \cite{Jointcaching}, the authors consider joint caching, routing, and channel assignment for video delivery over coordinated small-cell cellular systems of the future internet and utilize the column generation method to maximize the throughput of the system. The authors in \cite{ContextAware} allow jointly exploiting the wireless and social context of wireless users for optimizing the overall resources allocation and improving the traffic offload in small cell networks with device-to-device communication. 
{In \cite{Sung2016Efficient}, the authors propose an efficient cache placement strategy which uses separate channels for content dissemination and content service. The authors in \cite{kang2014mobile} propose a low-complexity search algorithm to minimize the average caching failure rate.}
However, most existing works on caching such as  \cite{Cluster,Cooperative, MeanField,Content,BackhaulChen,Tran2016Octopus,Cachingimprovement,Jointcaching,ContextAware} have focused on the performance analysis and simple caching approaches that may not scale well in a dense, content-centric CRAN. 
{Moreover, the existing cache replacement works \cite{Sung2016Efficient,kang2014mobile,De2011Optimum} which focus on wired CDNs do not consider the cache replacement in a wireless network such as CRANs in which one must investigate new caching challenges that stem from the dynamic and wireless nature of the system and from the backhaul and fronthaul limitations.}
In addition, these works assume a known content distribution that is then used to design an effective caching algorithm and, as such, they do not consider a proactive caching algorithm that can predict the content request distribution of each user. Finally, most of these existing works neglect the effect of the users' mobility. For updating the cached content, if one can make use of the long-term statistics of user mobility to predict the user association, the efficiency of content caching will be significantly improved \cite{Mobilityaware}. For proactive caching, the users' future position can also enable seamless handover and content download for users.

More recently, there has been significant interest in studying how prediction can be used for proactive caching such as in \cite{Bigdata,SoysaPredicting,NagarajaCaching,TadrousOn,ManytoMany,pompili2016elastic}.
%In \cite{Content}, the authors consider the dynamic content-centric RRHs clustering and multicast beamforming under the assumption that the content placement is assumed to be given. 
The authors in \cite{Bigdata} develop a data extraction method using the Hadoop platform to predict content popularity. The work in \cite{SoysaPredicting} proposes a fast threshold spread model to predict the future access pattern of multi-media content based on the social information. In \cite{NagarajaCaching}, the authors exploit the instantaneous demands of the users to estimate the content popularity and devise an optimal random caching strategy. In \cite{TadrousOn}, the authors derive bounds on the minimum possible cost achieved by any proactive caching policy and propose specific proactive caching strategies based on the cost function. In \cite{ManytoMany}, the authors formulate a caching problem as a many-to-many matching game to reduce the backhaul load and transmission delay. {The authors in \cite{pompili2016elastic} study the benefits of proactive operation but they develop any analytically rigorous learning technique to predict the users' behavior.} Despite these promising results, existing works such as \cite{Bigdata,SoysaPredicting,NagarajaCaching,TadrousOn,ManytoMany} do not take into account user-centric features, such as the demographics and user mobility. Moreover, such works cannot deal with massive volumes of data that stem from thousands of users connected to the BBUs of a CRAN, since they were developed for small-scale networks in which all processing is done at base station level.
Meanwhile, none of these works in \cite{Bigdata,SoysaPredicting,NagarajaCaching,TadrousOn,ManytoMany} analyzed the potential of using machine learning tools such as neural network for content prediction with mobility in a CRAN.  %The centralized nature of a CRAN provides an ideal setting for deploying machine learning that can predict the content request distribution for each user.

%Prediction presents many challenges in terms of caching content  \cite{Content,Onoptimal,Bigdata}.  
 
%In \cite{Content}, the authors consider the dynamic content-centric RRHs clustering and multicast beamforming under the assumption that the content placement is assumed to be given. A proactive caching procedure using perfect knowledge of content popularity is studied in \cite{Onoptimal}. The authors in \cite{Bigdata} develop a data extraction method using Hadoop platform to predict content popularity for each user. 

%In [7], a novel performance analysis that accounts for system-level dynamics is performed and an enabling architecture that captures the tight interaction between different radio access technologies is proposed. 
%However, most existing works \cite{Content,Onoptimal,Bigdata} on prediction of content popularity have focused on how to use the content prediction to optimize performance. Indeed, none of these works analyzed the potential of machine learning algorithm in prediction. Machine learning provides an ideal setting to extract and predict content information, since beyond the C-RAN architecture, the capacity of obtaining and handling all users' information increase the capacibility of learning.

The main contribution of this paper is a novel proactive caching framework that can accurately predict both the content request distribution and mobility pattern of each user and, subsequently, cache the most suitable contents while minimizing traffic and delay within a CRAN. The proposed approach enables the BBUs to dynamically learn and decide on which content to cache at the BBUs and RRHs, and how to cluster RRHs depending on the prediction of the users' content request distributions and their mobility patterns. 
%Within the context of a CRAN, developing such a dynamic caching and RRH clustering algorithm requires a dynamic approach so as to predict the content request distribution according to the change of environment.
Unlike previous studies such as \cite{Content}, \cite{Bigdata} and \cite{TadrousOn}, which require full knowledge of the users' content request distributions, we propose a novel approach to perform proactive content caching based on the powerful frameworks of echo state networks (ESNs) and sublinear algorithms \cite{Sublinear}. The use of ESNs enables the BBUs to quickly learn the distributions of users' content requests and locations without requiring the entire knowledge of the users' content requests. The entire knowledge of the user's content request is defined as the user's \emph{context} which includes the information about content request such as age, job, and location. The user's context significantly influence the user's content request distribution. Based on these predictions, the BBUs can determine which contents to cache at cloud cache and RRH cache and then offload the traffic. Moreover, the proposed sublinear approach enables the BBUs to quickly calculate the percentage of each content and determine the contents to cache without the need to scan all users' content request distributions.
To our best knowledge, beyond our work in \cite{Chen2016Echo} that applied ESN for LTE-U resource allocation, no work has studied the use of ESN for proactive caching. In order to evaluate the actual performance of the proposed approach, we use \emph{real data from Youku} for content simulations and use the \emph{realistic measured mobility data from the Beijing University of Posts and Telecommunications} for mobility simulations. Simulation results show that the proposed approach yields significant gains, in terms of the total effective capacity, that reach up to $27.8\%$ and $30.7\%$, respectively, compared to random caching with clustering and random caching without clustering. {Our key contributions are therefore:

\begin{itemize}
\item A novel proactive caching framework that can accurately predict both the content request distribution and mobility pattern of each user and, subsequently, cache the most suitable contents while minimizing traffic and delay within a CRAN.
\item A new ESN-based learning algorithm to predict the users' content request distribution and mobility patterns using users' contexts. 
\item Fundamental analysis on the memory capacity of the ESN with mobility data.     
\item A low-complexity sublinear algorithm that can quickly determine the RRHs clustering and which contents to store at RRH cache and cloud cache.   
\end{itemize} }   
 
The rest of this paper is organized as follows. The system model is described in Section \uppercase\expandafter{\romannumeral2}. The ESN-based content prediction approach is proposed in Section \uppercase\expandafter{\romannumeral3}. The proposed sublinear approach for content caching and RRH clustering is presented in Section \uppercase\expandafter{\romannumeral4}. In Section \uppercase\expandafter{\romannumeral5}, simulation results are  analyzed. Finally, conclusions are drawn in Section \uppercase\expandafter{\romannumeral6}. 

\vspace{-0cm}
\section{System Model and Problem Formulation}
\label{sec:SM}
\begin{figure}[!t]
  \begin{center}
   \vspace{0cm}
    \includegraphics[width=7cm]{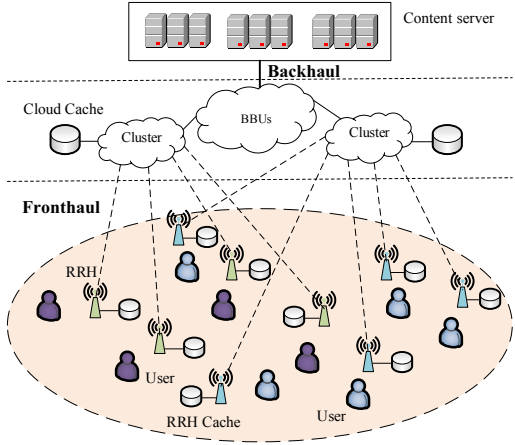}
    \vspace{-0.3cm}
    \caption{\label{CRAN} A CRAN using clustering and caching.}
  \end{center}\vspace{-0.8cm}
\end{figure}
Consider the downlink transmission of a cache-enabled CRAN in which a set $\mathcal{U} = \left\{ {1,2, \cdots ,U} \right\}$ of $U$ users are served by a set $\mathcal{R} = \{1,2,\ldots,{R}\}$ of $R$ RRHs. The RRHs are connected to the cloud pool of the BBUs via capacity-constrained, digital subscriber line (DSL) fronthaul links. The capacity of the fronthaul link is limited and $v_F$ represents the maximum fronthaul transmission rate for all users. As shown in Fig. 1, RRHs which have the same content request distributions are grouped into a virtual cluster which belongs to a set $\mathcal{M} = \mathcal{M}_1 \cup \ldots \cup \mathcal{M}_{M}$ of $M$ virtual clusters. {We assume that each user will always connect to its nearest RRHs cluster and can request at most one content at each time slot $\tau$.} The virtual clusters with their associated users allow the CRAN to use zero-forcing dirty paper coding (ZF-DPC) of multiple-input multiple-output (MIMO) systems to eliminate cluster interference. The proposed approach for forming virtual clusters is detailed in Section \ref{sectionS}. Virtual clusters are connected to the content servers via capacity-constrained wired backhaul links such as DSL. The capacity of the backhaul link is limited with $v_B$ being the maximum backhaul transmission rate for all users \cite{Sparsebeamforming}. Since each RRH may associate with more than one user, the RRH may have more than one type of content request distribution and belong to more than one cluster. Here, we note that the proposed approach can be deployed in any CRAN, irrespective of the way in which the functions are split between RRHs and BBUs. 

\subsection{Mobility Model}
In our model, the users can be mobile and have periodic mobility patterns. In particular, we consider a system in which each user will regularly visit a certain location. For example, certain users will often go to the same office for work at the same time during weekdays. We consider daily periodic mobility of users, which is collected once every $H$ time slots. The proposed approach for predicting the users' periodic mobility patterns is detailed in Section \ref{se:Mobility}. In our model, each user is assumed to be moving from the current location to a target location at a constant speed and this user will seamlessly switch to the nearest RRH as it moves. We ignore the RRH handover time duration that a user needs to transfer from one RRH to another.

Given each user's periodic mobility, we consider the caching of content, separately, at the RRHs and cloud. Caching at the cloud allows to offload the backhaul traffic and overcome the backhaul capacity limitations. In particular, the cloud cache can store the popular contents that all users request from the content servers thus alleviating the backhaul traffic and improve the transmission QoS. Caching at the RRH, referred to as RRH cache hereinafter, will only store the popular content that the associated users request. The RRH cache can significantly offload the traffic and reduce the transmission delay of both the fronthaul and backhaul. We assume that each content can be transmitted to a given user during time slot $\tau$. {In our model, a time slot represents the time duration during which each user has an invariant content request distribution. During each time slot, each user can receive several contents.} The RRH cache is updated each time slot $\tau$ and the cloud cache is updated during $T_\tau$ time slots. {We assume that the cached content update of each RRH depends only on the users located nearest to this RRH.} We also assume that the content server stores a set ${\mathcal{N}}=\{1, 2,\ldots, {N}\}$ of all contents required by all users. All contents are of equal size $L$. The set of $C_c$ cloud cache storage units is given by $\mathcal{C}_c=\{1,2,\cdots,{C_c}\}$, where $C_c \le N$. The set of $C_r$ RRH cache storage units is given by $\mathcal{C}_{r}=\{1,2,\cdots,{C_r}\}$, where $C_{r} \le N$, $r \in \mathcal{R}$. %Since each cache storage unit in the RRH and cloud stores a single, unique content, we have ${i}\bigcap {{j}}  = \emptyset, i \ne j, i,j \in {{\mathcal{C}_c}}$ or $i,j \in {{\mathcal{C}_r}}$, and ${\mathcal{C}_c}$, ${\mathcal{C}_{r}}\subseteq {\mathcal{S}}$. Note that the cache storage units at different RRHs can store the same contents.  

\subsection{Transmission Model}
\begin{figure}[!t]
  \begin{center}
   \vspace{0cm}
    \includegraphics[width=9.5cm]{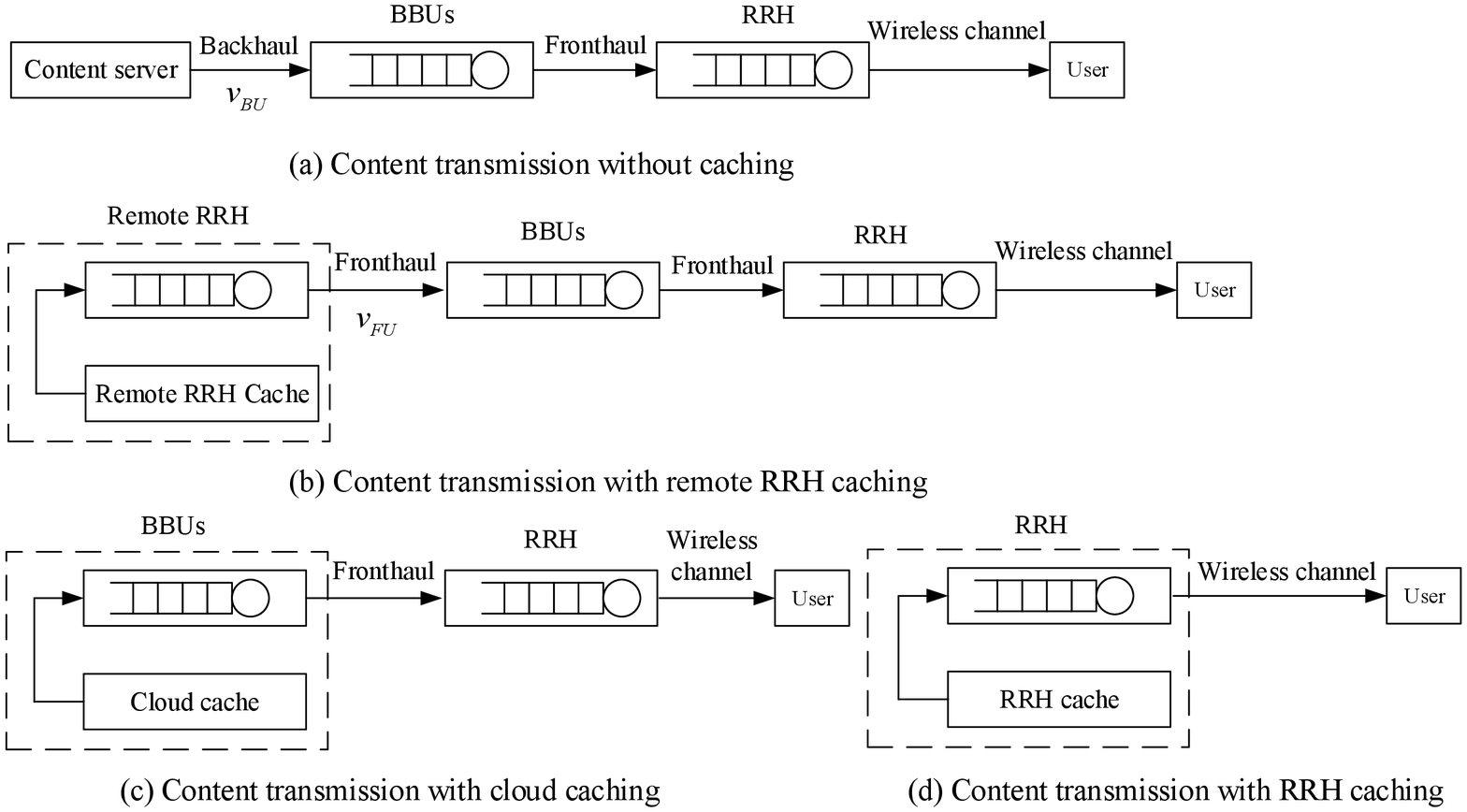}
    \vspace{-0.3cm}
    \caption{\label{CRAN} Content transmission in CRANs.}
  \end{center}\vspace{-0.8cm}
\end{figure}
As shown in Fig. 2, contents can be sent from: a) a content server, b) a remote RRH cache storage unit, c) a cloud cache storage unit, or d) an RRH cache storage unit to the user. An \emph{RRH} refers to an RRH that the user is already associated with, while a \emph{remote RRH} refers to other RRHs that store the user's required content but are not associated to this user. We assume that each content can be transmitted independently, and different contents are processed at different queues. The transmission rate of each content, $v_{BU}$, from the content server to the BBUs is:
\begin{equation}\label{eq:vBU}
\setlength{\abovedisplayskip}{0 pt}
\setlength{\belowdisplayskip}{0 pt}
{v_{BU}} =\frac{{{v_B}}}{{N_B}},
\end{equation}
where $N_B$ is the number of the users that request the contents that must be transmitted from the backhaul to the BBUs. Since the content transmission rates, from the cloud cache to the BBUs and from the RRH cache to the local RRH, can occur at a rate that is higher than that of the backhaul and fronthaul links such as in \cite{Cluster} and \cite{Cooperative}, we ignore the delay and QoS loss of these links. After transmitting the content to the BBUs, the content is delivered to the RRHs over fronthaul links. We also assume that the transmission rate from the RRH to the BBUs is the same as the rate from the BBUs to the RRH. Subsequently, the transmission rate, $v_{FU}$, of each content from the BBUs to the RRHs is
${v_{FU}} =\frac{{{v_F}}}{{ {{N_{F}}}}}$,
where $N_{F}$ is the number of the users that request contents that must be transmitted from the fronthaul to the RRHs. After transmitting the content to the RRHs, the content is transmitted to the users over the radio access channels. Therefore, the total transmission link of a specific content consists of one of the following links: a) content server-BBUs-RRH-user, b) cloud cache-BBUs-RRH-user, c) RRH cache-RRH-user, and d) remote RRH cache-remote RRH-BBUs-RRH-user. Note that the wireless link is time-varying due to the channel as opposed to the static, wired, DSL fronthaul and backhaul links.       
To mitigate interference, the RRHs can be clustered based on the content requests to leverage MIMO techniques. This, in turn, can also increase the effective capacity for each user, since the RRHs can cooperate and use ZF-DPC to transmit their data to the users. Therefore, the received signal-to-interference-plus-noise ratio of user $i$ from the nearest RRH $k \in \mathcal{M}_i$ at time $t$ is \cite{Exploring}:  
\begin{equation}
{\gamma _{t,ik}} = \frac{{Pd_{t,ik}^{ - \beta }{{\left\| {{h_{t,ik}}} \right\|}^2}}}{{\sum\limits_{j \in {\mathcal{M} \mathord{\left/
 {\vphantom {\mathcal{M} {{\mathcal{M}_i}}}} \right.
 \kern-\nulldelimiterspace} {{\mathcal{M}_i}}}} {Pd_{t,ij}^{ - \beta }{{\left\| {{h_{t,ij}}} \right\|}^2}+{\sigma ^2}} }},
\end{equation}
where $h_{t,ik}$ is the Rayleigh fading parameter and $d_{t,ik}^{-\beta}$ is the path loss at time $t$, with $d_{t,ik}$ being the distance between RRH $k$ and user $i$ at time $t$, and $\beta$ being the path loss exponent. $\sigma^2$ is the power of the Gaussian noise, and $P$ is the transmit power of each RRH, assumed to be equal for all RRHs. We also assume that the bandwidth of each downlink user is $B$. Since the user is moving and the distance between the RRH and user is varying, the channel capacity between RRH $k$ and user $i$ at time $t$ will be ${C_{t,ik}} =B{\log _2}\left( {1 + \gamma_{t,ik} } \right)$.
Since each user is served by the nearest RRH, we use $d_{t,i}$, $h_{t,i}$, $C_{t,i}$ and $\gamma_{t,i}$ to refer to $d_{t,ik}$, $h_{t,ik}$ $C_{t,ik}$ and $\gamma_{t,ik}$, for simplicity.  {Note that, ZF-DPC is implemented in the cloud and can be used for any transmission link.}     

\subsection{Effective Capacity}
Since the capacity $C_{t,i}$ does not account for delay, it cannot characterize the QoS of requested content. In contrast, the notion of an effective capacity, as defined in \cite{Effective}, represents a useful metric to capture the maximum content transmission rate of a channel with a specific QoS guarantee. First, we introduce the notion of a QoS exponent that allows quantifying the QoS of a requested content and, then, we define the effective capacity. The QoS exponent related to the transmission of a given content $n$ to a user $i$ with a stochastic waiting queue length $Q_{i,n}$ is \cite{Effective}: 
\begin{equation}\label{eq:thetail}
\setlength{\abovedisplayskip}{4 pt}
\theta_{i,n}  = \mathop {\lim }\limits_{q \to \infty } \frac{{\log_2 \Pr \left[ {Q_{i,n} > q} \right]}}{q},
\end{equation}
where $q$ is the system allowable threshold of queue length. For a large threshold value ${{q_{\max }}}$, the buffer violation probability of content $n$ for user $i$ can be approximated by:
\begin{equation}
\setlength{\abovedisplayskip}{3 pt}
\setlength{\belowdisplayskip}{3 pt}
\Pr \left[ {Q_{i,n} > {q_{\max }}} \right] \mathop  \approx {e^{ - \theta_{i,n} {q_{\max }}}}.
\end{equation}
This approximation is obtained from the large deviation theory. Then, the relation between buffer violation probability and delay violation probability for user $i$ with content $n$ is \cite{Effective}:
\begin{equation}
\setlength{\abovedisplayskip}{4 pt}
\setlength{\belowdisplayskip}{4 pt}
 \Pr \left[ {D_{i,n} > {D_{\max }}} \right] \le k\sqrt {\Pr \left[{Q_{i,n} > {q_{\max }}} \right]}, 
 \end{equation}
 where $D_{i,n}$ is the delay of transmitting content $n$ to user $i$ and $D_{\max}$ is the maximum tolerable delay of each content transmission. Here, $k$ is a positive constant and the maximum delay $q_{\max}=cD_{\max}$, with $c$ being the transmission rate over the transmission links. Therefore, $\theta_{i,n}$ can be treated as the QoS exponent of user $i$ transmitting content $n$ which also represents the delay constraint. A smaller $\theta_{i,n}$ reflects a looser QoS requirement, while a larger $\theta_{i,n}$ expresses a more strict QoS requirement. The QoS exponent pertaining to the transmission of a content $n$ to user $i$ with delay $D_{i,n}$ is \cite{Cluster}:
 \begin{equation}\label{eq:thetaD}
 \setlength{\abovedisplayskip}{4 pt}
\setlength{\belowdisplayskip}{4 pt}
 \theta_{i,n}  = \mathop {\lim }\limits_{{D_{\max }} \to \infty } \frac{{-\log \Pr\left( {D_{i,n} > {D_{\max }}} \right)}}{{{D_{\max }} - {{{N_h}L} \mathord{\left/
 {\vphantom {{{N_h}L} v}} \right.
 \kern-\nulldelimiterspace} v}}},
\end{equation}
where $N_h$ indicates the number of hops of each transmission link and $v$ indicates the rate over the wired fronthaul and backhaul links. Based on (\ref{eq:thetail})-(\ref{eq:thetaD}),  the cumulative distribution function of delay of user $i$ transmitting content $n$ with a delay threshold $D_{\max}$ is given by:
 \begin{equation}\label{eq:PrD}
 \setlength{\abovedisplayskip}{4 pt}
\setlength{\belowdisplayskip}{4 pt}
\Pr \left( {D_{i,n} > {D_{\max }}} \right) \approx {e^{ - \theta_{i,n} \left( {{D_{\max }} - {{{N_h}} \mathord{\left/
 {\vphantom {{{N_h}} v}} \right.
 \kern-\nulldelimiterspace} v}} \right)}}.
\end{equation}
%\vspace{-0cm}
The corresponding \emph{QoS exponents} pertaining to the transmission of a content $n$ to a user $i$ can be given as follows: a) content server-BBUs-RRH-user $\theta_{i,n}^{S}$, b) cloud cache-BBUs-RRH-user $\theta_{i,n}^{A}$, c) local RRH cache-RRH-user $\theta_{i,n}^{O}$, d) remote RRH cache-remote RRH-BBUs-RRH-user $\theta_{i,n}^{G}$. Since the QoS of each link depends on the QoS exponents, we use the relationship between the QoS exponent parameters to represent the transmission quality of each link. In order to quantify the relationship of the QoS exponents among these links, we state the following result:

\begin{proposition}\label{pro1}
\emph{To achieve the same QoS and delay of transmitting content $n$ over the wired fronthaul and backhaul links, the QoS exponents of the four transmission links of content $n$ with $v_{BU}$ and $v_{FU}$ must satisfy the following conditions:}
\begin{equation*}
\setlength{\abovedisplayskip}{3 pt}
\setlength{\belowdisplayskip}{3 pt}
\begin{split}
&\text{a)}\;\;\theta _{i,n}^S = \frac{{\theta _{i,n}^O}}{{1 - {{2L} \mathord{\left/
 {\vphantom {{2L} {{v_{BU}}{D_{\max }}}}} \right.
 \kern-\nulldelimiterspace} {{v_{BU}}{D_{\max }}}}}}, \;\;
\text{b)}\;\;\theta _{i,n}^A = \frac{{\theta _{i,n}^O}}{{1 - {{L} \mathord{\left/
 {\vphantom {{2L} {{v_{FU}}{D_{\max }}}}} \right.
 \kern-\nulldelimiterspace} {{v_{FU}}{D_{\max }}}}}},\;\;\\
&\text{c)}\;\;\theta _{i,n}^G = \frac{{\theta _{i,n}^O}}{{1 - {{2L} \mathord{\left/
 {\vphantom {{3L} {{v_{FU}}{D_{\max }}}}} \right.
 \kern-\nulldelimiterspace} {{v_{FU}}{D_{\max }}}}}}.
\end{split}
\end{equation*}
\end{proposition}
\begin{proof} See Appendix A.
\end{proof}\vspace*{-0cm}

Proposition \ref{pro1} captures the relationship between the QoS exponents of different links. This relationship indicates the transmission QoS for each link. From Proposition \ref{pro1}, we can see that, given the QoS requirement $\theta_{i,n}^{O}$ for transmitting content $n$, the only way to satisfy the QoS requirement $\theta_{i,n}^{O}$ over a link b) is to take the limits of the transmission rate $v_{FU}$ to infinity. Based on Proposition \ref{pro1} and $\theta_{i,n}^O$, we can compute the QoS exponents achieved by the transmission of a content $n$ from different links. The BBUs can select an appropriate link for each content transmission with a QoS guarantee according to the QoS exponent of each link.

Given these basic definitions, the effective capacity of each user is given next. Since the speed of each moving user is constant, the cumulative channel capacity during the time slot $\tau$ is given as ${C_{\tau ,i}} = \sum\nolimits_{t = 1,2, \ldots ,\tau } {{C_{t,i}}}  = {\mathbb{E}_{d_i,h_i}}[{C_{t,i}}]$. Therefore, the effective capacity of user $i$ {receiving} content $n$ during time $\tau$ is given by \cite{Effective}: 
%Consequently, effective capacity of $\theta _l^j$ for user $i$ can be given by \cite{Effective}:
%\begin{equation}\label{eq:E}
%E\left( \theta _l^j, d_i  \right) =  - \mathop {\lim }\limits_{t \to \infty } \frac{1}{{\theta_l^j t}}\log_2 \mathbb{E}\left\{ {{e^{ - \theta _l^j S_{i}\left( t \right)}}} \right\},
%\end{equation}
%where $j \in \left\{ {S,H} \right\}$ and $S_{i}(t)$ is the transmitted content accumulated on time domain. We assume that the channel condition is block fading in which the channel coefficient is a constant during a time unit $\tau$. In each time $\tau$, the effective capacity can be further derived as follows:
\begin{equation}\label{eq:E}
\setlength{\abovedisplayskip}{4 pt}
\setlength{\belowdisplayskip}{4 pt}
E_{\tau,i}\!\left(  \theta _{i,n_{i\tau},\tau}^j\right)\!=\!  - \mathop \frac{1}{{ \theta _{i,n_{i\tau},\tau}^j \tau}}\log_2 \mathbb{E}_{d_i,h_i}\!\!\!\left[ {{e^{ -  \theta _{i,n_{i\tau},\tau}^j C_{\tau,i} }}} \right],
%\mathds{1}_{\{\tau \ne T_h \} },
\end{equation} 
where $n_{i\tau}$ represents the content that user $i$ requests at time slot $\tau$, $j \in \left\{ {O,A,S,G} \right\}$ indicates the link that transmits the content $n$ to user $i$ 
%$\mathds{1}_{\{\text{condition}\}}$ equals 1 (0) when condition is true (false), 
and $\mathbb{E}_{d_i,h_i}\left[x\right]$ is the expectation of $x$ with respect to distribution of $d_i$ and $h_i$.
%Given the set $\mathcal{U}_l$ of users that require content $s_l$ in micro cloud $m_l$, 
Based on (\ref{eq:E}), the sum effective capacity of all moving users during time slot $k$ is:
\begin{equation}\label{eq:Se}
\setlength{\abovedisplayskip}{4 pt}
\setlength{\belowdisplayskip}{4 pt}
{E_k} = \sum\limits_{i \in \mathcal{U}} {{E_{k,i}}\left( {\theta _{i,n_{ik},k}^j} \right)}.
\end{equation}
The sum effective capacity $\bar E$ is analyzed during $T$ time slots. Therefore, the long term effective capacity $\bar E$ is given by $\bar E = \frac{1}{T}\sum\nolimits_{k = 1}^T {{E_k}} $. $\bar E$ actually captures the delay and QoS of contents that are transmitted from the content server, remote RRHs, and caches to the network users during a period $T$. {Note that the use of the effective capacity is known to be valid, as long as the following two conditions hold \cite{Effective}: a) Each user's content transmission has its own, individual queue. b) The buffer of each queue is of infinite (large) size.  
Since the BBUs will allocate separate spectrum resource for each user's requested content transmission, we can consider that each users' content transmission is independent and hence, condition a) is satisfied. For condition 2), since we deal with the queue of each user at the level of a cloud-based system, such an assumption will be reasonable, given the high capabilities of a cloud server. Therefore, the conditions are applicable to the content transmission scenario in the proposed framework.}    
\subsection{Problem Formulation}
Given this system model, our goal is to develop an effective caching scheme and content RRH clustering approach to reduce the interference and offload the traffic of the backhaul and fronthaul based on the predictions of the users' content request distributions and periodic mobility patterns. To achieve this goal, we formulate a QoS and delay optimization problem whose objective is to maximize the long-term sum effective capacity. This optimization problem of caching involves predicting the content request distribution and periodic location for each user, and finding optimal contents to cache at the BBUs and RRHs. This problem can be formulated as follows:
\vspace{-0.05cm}
\addtocounter{equation}{0}
\setlength{\abovedisplayskip}{5 pt}
\setlength{\belowdisplayskip}{0 pt}
\begin{equation}\label{eq:sum}
\addtolength\abovedisplayshortskip{-8pt}
 \addtolength\belowdisplayshortskip{-7pt}
\!\!\!\!\!\!\!\!\!\!\!\!\!\!\!\!\!\!\!\!\!\!\mathop {\max }\limits_{{\mathcal{C}_c, \mathcal{C}_{r}}} {\bar E} =\mathop {\max }\limits_{{\mathcal{C}_c, \mathcal{C}_{r}}} \frac{1}{{{T}}}\sum\limits_{k = 1}^{{T}} {\sum\limits_{i \in \mathcal{U}} {E_{k,i}\left( {\theta _{i,n_{ik},k}^j} \right)} } ,\\
\end{equation}
\vspace{-0.1cm}
\begin{align}\label{c1}
&\;\;\;\;\;\;\;\;\;\;\;\;\;\;\;\;\;\;\;\;\;\;\;\;\;\;\;\!\!\!\!\!\!\!\!\!\!\!\!\!\!\!\!\!\!\!\!\!\!\!\!\!\!\!\!\!\!\text{s. t.}\scalebox{1}{$\;\;\;\;{m}\bigcap {{f}}  = \emptyset, m \ne f, m, f \in {{\mathcal{C}_c}}$, or $m, f \in {{\mathcal{C}_r}}$,} \tag{\theequation a}\\
&\;\;\;\;\;\;\;\;\;\;\;\;\;\;\;\;\;\;\;\;\;\;\;\;\!\!\!\!\!\!\!\!\!\!\!\!\!\!\!\!\!\!\!\!\!\!\!\scalebox{1}{$\;\;\;\;\;\;\;\;j \in \left\{ {O,A,S,G} \right\}$,} \tag{\theequation b}\\
&\scalebox{1}{$\;\;\;\;\;\;\;\;\;\;\;\;\;\;\;\;\;\;\;\;\;\;\;\;\;\;\!\!\!\!\!\!\!\!\!\!\!\!\!\mathcal{C}_c, \mathcal{C}_{r}, n_{ik}, \subseteq \mathcal{N}, r \in \mathcal{R},\;\;\;\;$} \tag{\theequation c}
%$L = \sum\nolimits_i {{a_i}}  + \sum\nolimits_i {{b_i}}$
%&\scalebox{1}{$\;\;\;\;\;\;\;\;\;\;\;\;{x_{ij}} \in \left\{ {0,1} \right\},\;\;\;\;\;\forall i \in \mathcal{U}, j \in \mathcal{B}  $}, \tag{\theequation e}\\
%&\scalebox{1}{$\;\;\;\;\;\;\;\;\;{G} \le {S}$} \tag{\theequation c},
\end{align}
where $\mathcal{C}_c$ and $\mathcal{C}_{r}$ represent, respectively, {the set of contents that stored in the cloud cache and RRH cache}, (\ref{eq:sum}a) captures the fact that each cache storage unit in the RRH and cloud stores a single, unique content, (\ref{eq:sum}b) represents that the links selection of transmitting each content, and (\ref{eq:sum}c) indicates that the contents at the cache will all come from the content server. Here, we note that, storing contents in the cache can increase the rates $v_{BU}$ and $v_{FU}$ of the backhaul and fronthaul which, in turn, results in the increase of the effective capacity. Moreover, storing the most popular contents in the cache can maximize the number of users receiving content from the cache. This, in turn, will lead to maximizing the total effective capacity. Meanwhile, the prediction of each user's mobility pattern can be combined with the prediction of the user's content request distribution to determine which content to store in which RRH cache. Such intelligent caching will, in turn, result in the increase of the effective capacity. Finally, RRHs' clustering with MIMO is used to further improve the effective capacity by mitigating interference within each cluster. Fig. \ref{solution} summarizes the proposed framework that is used to solve the problem in (\ref{eq:sum}). Within this framework, we first use the ESNs predictions of content request distribution and mobility pattern to calculate the average content request percentage for each RRH's associated users. Based on the RRH's average content request percentage, the BBUs determine the content that must be cached at each RRH. Based on the RRH caching and the content request distribution of each user, the BBUs will then decide on which content to cache at cloud.  

\begin{figure}[!t]
  \begin{center}
   \vspace{0cm}
    \includegraphics[width=9cm]{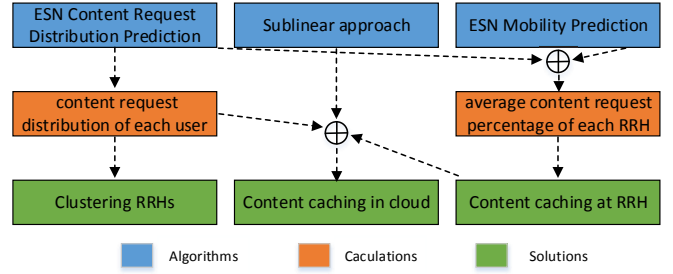}
    \vspace{-0.4cm}
    \caption{\label{solution} Overview of the problem solution.}
  \end{center}\vspace{-0.9cm}
\end{figure}

\section{Echo State Networks for Content Prediction and Mobility}\label{section2} 
The optimization problem in (\ref{eq:sum}) is challenging to solve, because the effective capacity depends on the prediction of the content request distribution which determines the popularity of a given content. The effective capacity also depends on the prediction of the user's mobility pattern that will determine the user association thus affecting the RRH caching. In fact, since the RRH caching and cloud caching need to be aware of the content request distribution of each user in advance, the optimization problem is difficult to solve using conventional optimization algorithms since such conventional approaches are not able to predict the user's content request distribution for the BBUs. Moreover, in a dense CRAN, the BBUs may not have the entire knowledge of the users' contexts that are needed to  %The entire knowledge of the users' content requests is defined as the users' \emph{contexts} which includes the information about content request such as age, job, and location. The users' contexts significantly influence the users' content request distributions. 
improve the accuracy of the content and mobility predictions thus affecting the cache placement strategy. These reasons render the optimization problem in (\ref{eq:sum}) challenging to solve in the presence of limited information. To address these challenges, we propose a novel approach to predict the content request distribution and mobility pattern for each user based on the powerful framework of \emph{echo state networks} \cite{Harnessing}. ESNs are an emerging type of recurrent neural networks \cite{APractical} that can track the state of a network and predict the future information, such as content request distribution and user mobility pattern, over time.    
\subsection{Content Distribution Prediction}
In this subsection, we formulate the ESN-based content request distribution prediction algorithm. A prediction approach based on ESNs consists of four components: a) agents, b) input, c) output, and d) ESN model. The ESN will allow us to build the content request distribution based on each user's context. The proposed ESN-based prediction approach is thus defined by the following key components:

$\bullet$ \emph{Agents}: The agents in our ESNs are the BBUs. Since each ESN scheme typically performs prediction for just one user, the BBUs must implement $U$ ESN algorithms at each time slot.  

$\bullet$ \emph{Input:} The ESN takes an input vector $\boldsymbol{x}_{t,j}=\left[ {x_{tj1}, \cdots , x_{tjK}} \right]^{\mathrm{T}}$ that represents the context of user $j$ at time $t$ including content request time, week, gender, occupation, age, and device type (e.g., tablet or smartphone). The vector $\boldsymbol{x}_{t,j}$ is then used to determine the content request distribution ${\boldsymbol{y} _{t,j}}$ for user $j$. For example, the types of videos and TV programs that interest young teenage students, will be significantly different from those that interest a much older demographic. Indeed, the various demographics and user information will be critical to determine the content request preferences of various users. Here, $K$ is the number of properties that constitute the context information of user $j$.

$\bullet$ \emph{Output:} The output of the ESN at time $t$ is a vector of probabilities $\boldsymbol{y}_{t,j}= \left[ {{p_{tj1}},{p_{tj2}}, \ldots ,{p_{tjN}}} \right]$ that represents the probability distribution of content request of user $j$, where $p_{tjn}$ is the probability that user $j$ requests content $n$ at time $t$. 

$\bullet$ \emph{ESN Model:} An ESN model can approximate the function between the input $\boldsymbol{x}_{t,j}$ and output $\boldsymbol{y}_{t,j}$, thus building the relationship between each user's context and the content request distribution. For each user $j$, an ESN model is essentially a dynamic neural network, known as the dynamic reservoir, which will be combined with the input $\boldsymbol{x}_{t,j}$ representing the context of user $j$. Mathematically, the dynamic reservoir consists of the input weight matrix $\boldsymbol{W}_j^{\alpha,in} \in {\mathbb{R}^{N_w \times K}}$, and the recurrent matrix $\boldsymbol{W}_j^\alpha \in {\mathbb{R}^{N_w \times N_w}}$, where $N_w$ is the number of the dynamic reservoir units that the BBUs use to store the context of user $j$. The output weight matrix $\boldsymbol{W}_j^{\alpha,out} \in {\mathbb{R}^{N \times \left(N_w+K\right)}}$ is trained to approximate the prediction function. $\boldsymbol{W}_j^{\alpha,out}$ essentially reflects the relationship between context and content request distribution for user $j$. The dynamic reservoir of user $j$ is therefore given by the pair $\left( \boldsymbol{W}_j^{\alpha,in}, \boldsymbol{W}_j^\alpha \right)$ which is initially generated randomly via a uniform distribution and $\boldsymbol{W}_j^\alpha$ is defined as a sparse matrix with a spectral radius less than one \cite{APractical}. $\boldsymbol{W}_j^{\alpha,out}$ is also initialized randomly via a uniform distribution. By training the output matrix $\boldsymbol{W}_j^{\alpha,out}$, the proposed ESN model can predict the content request distribution based on the input $\boldsymbol{x}_{t,j}$, which will then provide the samples for the sublinear algorithm in Section \ref{sectionS} that effectively determines which content to cache. %In this ESN model, one needs to only train $\boldsymbol{W}_j^{out}$ to approximate the reward function which illustrates that ESNs are easy to train [14]. Even though the dynamic reservoir is initially generated randomly, it will be combined with input to store the users' states and it will also be combined with the trained output matrix to approximate the reward function.

Given these basic definitions, we introduce the dynamic reservoir state ${\boldsymbol{v}_{t,j}^\alpha}$ of user $j$ at time $t$ which is used to store the states of user $j$ as follows: 
\begin{equation}\label{eq:state}
\setlength{\abovedisplayskip}{4 pt}
\setlength{\belowdisplayskip}{4 pt}
{\boldsymbol{v}_{t,j}^\alpha} ={\mathop{f}\nolimits}\!\left( {\boldsymbol{W}_j^\alpha{\boldsymbol{v}_{t - 1,j}^\alpha} + \boldsymbol{W}_j^{\alpha,in}{\boldsymbol{x}_{t,j}}} \right),
\end{equation}
where $f\left(  \cdot  \right)$ is the tanh function. Suppose that each user $j$ has a content request at each time slot. Then, the proposed ESN model will output a vector that captures the content request distribution of user $j$ at time $t$. The output yields the actual distribution of content request at time $t$:
\begin{equation}\label{eq:es}
\setlength{\abovedisplayskip}{4 pt}
\setlength{\belowdisplayskip}{4 pt}
\boldsymbol{y}_{t,j} \left(\boldsymbol{x}_{t,j}\right) = {\boldsymbol{W}_{t,j}^{\alpha,out}}\left[ {{\boldsymbol{v}_{t,j}^\alpha};{\boldsymbol{x}_{t,j}}} \right],
\end{equation}
%After we obtain the probability distribution of content request of user $j$, we select the content which has the highest probability as the exact content request in the next time. If more than one contents have the highest probability, we select both as the exact content requests. 
where ${\boldsymbol{W}_{t,j}^{\alpha,out}}$ is output matrix ${\boldsymbol{W}_{j}^{\alpha,out}}$ at time $t$. In other words, (\ref{eq:es}) is used to build the relationship between input ${\boldsymbol{x}_{t,j}}$ and the output $\boldsymbol{y}_{t,j}$. In order to build this relationship, we need to train $\boldsymbol{W}_{t,j}^{\alpha,out}$. A linear gradient descent approach is used to derive the following update rule,
\begin{equation}\label{eq:updatew}
\setlength{\abovedisplayskip}{4 pt}
\setlength{\belowdisplayskip}{4 pt}
{\boldsymbol{W}_{t + 1,j}^{\alpha,out}} = {\boldsymbol{W}_{t,j}^{\alpha,out}} + {\lambda^\alpha }\left( {\boldsymbol{e}_{t,j}-\boldsymbol{y}_{t,j}\left(\boldsymbol{x}_{t,j} \right)} \right)\left[ {{\boldsymbol{v}_{t,j}^\alpha};{\boldsymbol{x}_{t,j}}} \right]^{\mathrm{T}},
\end{equation}
where $\lambda^\alpha$ is the learning rate and $\boldsymbol{e}_{t,j}^\alpha$ is the real content request distribution of user $j$ at time $t$. Indeed, (\ref{eq:updatew}) shows how an ESN can approximate to the function of (\ref{eq:es}).   
\subsection{Mobility Prediction}\label{se:Mobility}
{In this subsection, we study the mobility pattern prediction of each user. First, in mobile networks, the locations of the users can provide key information on the user-RRH association to the content servers which can transmit the most popular contents to the corresponding RRHs. Second, the type of the content request will in fact depend
 on the users' locations.}
%In order to determine which content to cache at the RRHs, we need to predict the trajectory of each user. 
Therefore, we introduce a minimum complexity ESN algorithm to predict the user trajectory in this subsection. Unlike the ESN prediction algorithm of the previous subsection, the ESN prediction algorithm of the user mobility proposed here is based on the minimum complexity dynamic reservoir and adopts an offline method to train the output matrix. The main reason behind this is that the prediction of user mobility can be taken as a time series and needs more data to train the output matrix. Therefore, we use a low complexity ESN to train the output matrix and predict the position of each user. The ESN will help us predict the user's position based on the positions that the user had visited over a given past history, such as the past few weeks, for example. Here, the mobility prediction ESN will also include four components, with the BBUs being the agents, and the other components being:

$\bullet$ \emph{Input:}  ${m}_{t,j}$ represents the current location of user $j$. This input ${m}_{t,j}$ combining with the history input data, $\left[{m}_{t-1,j},\dots,{m}_{t-M,j}\right]$, determines the positions ${\boldsymbol{s} _{t,j}}$ that the user is expected to visit. Here, $M$ denotes the number of the history data that an ESN can record. 

$\bullet$ \emph{Output:} $\boldsymbol{s}_{t,j}=\left[ {s_{tj1}, \cdots , s_{tjN_s}} \right]^{\mathrm{T}}$ represents the position that user $j$ is predicted to visit for the next steps, where $N_s$ represents the number of position that user $j$ is expected to visit in the next $N_s$ time duration $H$. 

$\bullet$ \emph{Mobility Prediction ESN Model:} An ESN model builds the relationship between the user's context and positions that the user will visit. For each user $j$, an ESN model will be combined with the input $\boldsymbol{m}_{t,j}$ to record the position that the user has visited over a given past history. The ESN model consists of the input weight matrix $\boldsymbol{W}_j^{in} \in {\mathbb{R}^{W \times 1}}$, the recurrent matrix $\boldsymbol{W}_j \in {\mathbb{R}^{W \times W}}$, where $W$ is the number of units of the dynamic reservoir that the BBUs use to store position records of user $j$, and the output weight matrix $\boldsymbol{W}_j^{out} \in {\mathbb{R}^{N_s \times W}}$. The generation of $\boldsymbol{W}_j^{in}$ and $\boldsymbol{W}_j^{out}$ are similar to the content distribution prediction approach. $\boldsymbol{W}_j$ is defined as a full rank matrix defined as follows:
\begin{equation}\small
\setlength{\abovedisplayskip}{4 pt}
\setlength{\belowdisplayskip}{4 pt}
\boldsymbol{W}_j=\left[ {\begin{array}{*{20}{c}}
{{0}}&{{0}}& \cdots &{{w}}\\
{{w}}&0&0&0\\
0& \ddots &0&0\\
0&0&{{w}}&0
\end{array}} \right],
\end{equation}
where $w$ can be set as a constant or follows a distribution, such as uniform distribution. The value of $w$ will be detailed in Theorem \ref{theorem1}.
Given these basic definitions, we use a linear update method to update the dynamic reservoir state ${\boldsymbol{v}_{t,j}}$ of user $j$, which is used to record the positions that user $j$ has visited as follows: 
\begin{equation}\label{eq:reservoirstate}
\setlength{\abovedisplayskip}{4 pt}
\setlength{\belowdisplayskip}{4 pt}
{\boldsymbol{v}_{t,j}} ={\boldsymbol{W}_j{\boldsymbol{v}_{t - 1,j}} + \boldsymbol{W}_j^{in}{{m}_{t,j}}}.
\end{equation}
The position of output $\boldsymbol{s}_{t,j}$ based on ${\boldsymbol{v}_{t,j}}$ is given by:
\begin{equation}\label{eq:y2}
\setlength{\abovedisplayskip}{4 pt}
\setlength{\belowdisplayskip}{4 pt}
{\boldsymbol{s}_{t,j}} = \boldsymbol{W}_j^{out}{\boldsymbol{v}_{t,j}}.
\end{equation}
In contrast to (\ref{eq:updatew}), $\boldsymbol{W}_j^{out}$ of user $j$ is trained in an offline manner using ridge regression \cite{APractical}:
\begin{equation}\label{eq:w2}
\setlength{\abovedisplayskip}{4 pt}
\setlength{\belowdisplayskip}{4 pt}
 \boldsymbol{W}_j^{out}=\boldsymbol{s}_j{\boldsymbol{v}_j^{\rm T}}{\left({\boldsymbol{v}_j^{\rm T}}\boldsymbol{v}_j + {\lambda ^2}\boldsymbol{\rm I}\right)^{ - 1}},
 \end{equation}
where $\boldsymbol{v}_j=\left[\boldsymbol{v}_{1,j},\dots, \boldsymbol{v}_{N_{tr},j}\right] \in \mathbb{R}^{W \times N_{tr}} $ is the reservoir states of user $j$ for a period $N_{tr}$, $\boldsymbol{s}_j$ is the output during a period $N_{tr}$, and $\boldsymbol{\rm I}$ is the identity matrix.     

Given these basic definitions, we derive the memory capacity of the mobility ESN which is related to the number of reservoir units and the value of $w$ in $\boldsymbol{W}_j$. {The ESN memory capacity is used to quantify the number of the history input data that an ESN can record. For the prediction of the mobility pattern,} the memory capacity of the mobility ESN determines the ability of this model to record the locations that each user $j$ has visited. First, we define the following $K \times K$ matrix, given that ${\boldsymbol{W}_j^{in}} = {\left[ {w_1^{in}, \ldots ,w_{W}^{in}} \right]^{\rm T}}$:
\begin{equation}\small
\boldsymbol{\Omega}  = \left[ {\begin{array}{*{20}{c}}
{w_1^{in}}&{w_{W}^{in}}& \cdots &{w_2^{in}}\\
{w_2^{in}}&{w_1^{in}}& \cdots &{w_3^{in}}\\
 \vdots & \vdots & \cdots & \vdots \\
{w_{W}^{in}}&{w_{W-1}^{in}}& \cdots &{w_1^{in}}
\end{array}} \right].
\end{equation}
Then, the memory capacity of the mobility ESN can be given as follows: 
\vspace{-0.12cm}      
\begin{theorem}\label{theorem1}
\emph{In a mobility ESN, we assume that the reservoir $\boldsymbol{W}_j$ is generated randomly via a specified distribution, $\boldsymbol{W}_j^{in}$ guarantees that the matrix $\boldsymbol{\Omega}$ regular, and the input ${m}_{t,j}$ has periodicity. Then, the memory capacity of this mobility ESN will be given by:}
\begin{equation}\label{eq:theorem1}\small
\setlength{\abovedisplayskip}{4 pt}
\setlength{\belowdisplayskip}{4 pt}
M\!\! =\!\!\!\sum\limits_{k = 0}^{W - 1} {{{\!\!\left(\sum\limits_{j = 0}^\infty  {\mathbb{E}\!\left[{w^{2Wj + 2k}}\right]\!} \right)^{\!\! \!- 1}}}\!\!\sum\limits_{j = 0}^\infty  {\mathbb{E}{{\left[{w^{Wj + k}}\right]}^2}} } \!\!-{{\left(\sum\limits_{j = 0}^\infty {\mathbb{E}\left[{w^{2Wj}}\right]} \right)^{\!\!\!-1}}}\!\!.
\end{equation}

\end{theorem} 
\begin{proof} See Appendix B.
\end{proof}
The memory capacity of the mobility ESN indicates the ability of the mobility ESN model to record the locations that each user has visited. From Theorem \ref{theorem1}, we can see that the ESN memory capacity depends on the distribution of reservoir unit $w$ and the number of the reservoir units $W$. A larger memory capacity implies that the ESN can store more locations that the user has visited. The visited locations can improve the prediction of the user mobility. Since the size of the reservoir $\boldsymbol{W}_j$ and the value of $w$ will have an effect on the mobility prediction, we need to set the size of $\boldsymbol{W}_j$ appropriately to satisfy the memory capacity requirement of the user mobility prediction based on Theorem \ref{theorem1}.  
Different from the existing works in \cite{Short} and \cite{Minimum} that use an independent and identically distributed input stream to derive the ESN memory capacity, we formulate the ESN memory capacity with a periodic input stream. Next, we formulate the upper and lower bounds on the ESN memory capacity with different distributions of the reservoir $\boldsymbol{W}_j$. The upper bound of the ESN memory capacity can give a guidance for the design of $\boldsymbol{W}_j$. 

\begin{proposition}\label{pro3}
\emph{Given the distribution of the reservoir $\boldsymbol{W}_j$ $\left(\left| w \right| < 1\right)$, the upper and lower bounds of the memory capacity of the mobility ESNs are given by:}\\
\romannumeral1) \emph{ If $w \in \boldsymbol{W}_j$ follows a zero-mean distribution $\left(\text{i.e. } {w \in \left[ { - 1,1} \right]} \right)$, then $0 \le M < {\left\lfloor {\frac{W}{2}} \right\rfloor }+1$, where $\left\lfloor {x} \right\rfloor $ is the floor function of $x$.}\\
\romannumeral2) \emph{ If $w \in \boldsymbol{W}_j$ follows a distribution that makes $w > 0$, then $0 < M < W$.}
\end{proposition}

\begin{proof} See Appendix C.   
\end{proof}
From Proposition \ref{pro3}, we can see that, as $P\left(w = a\right) = 1$ and $a \to 1$, the memory capacity of the mobility ESN $M$ will be equal to the number of reservoir units $W$. Since we predict $N_s$ locations for each user at time $t$, we need to set the number of reservoir units above $W=N_s+1$.     
\section{Sublinear Algorithm for Caching}\label{sectionS} 
The predictions of the content request distribution and user mobility pattern in Section \ref{section2} must now be leveraged to determine which content to cache at the RRHs, cluster the RRHs at each time slot, and identify which contents to store in cloud cache during a given period. Clustering the RRHs based on the request content will also enable the CRAN to use ZF-DPC of MIMO to eliminate cluster interference. However, it is challenging for the BBUs to scan each content request distribution prediction among the thousands of users' content request distribution predictions resulting from the ESNs' output within a limited time. In addition, in a dense CRAN, the BBUs may not have the entire knowledge of the users' contexts and distributions of content request in a given period, thus making it challenging to determine which contents to cache as per (\ref{eq:sum}). To address these challenges, we propose a novel \emph{sublinear approach} for caching\cite{Sublinear}. %Since RRH clustering and caching are both based on the content request distributions, the proposed algorithm can solve both problems. 

A sublinear algorithm is typically developed based on random sampling theory and probability theory \cite{Sublinear} to perform effective big data analytics. In particular, sublinear approaches can obtain the approximation results to the optimal result of an optimization problems by only looking at a subset of the data for the case in which the total amount of data is so massive that even linear processing time is not affordable. For our model, a sublinear approach will enable the BBUs to compute the average of the content request percentage of all users so as to determine content caching at the cloud without scanning through the massive volume of data pertaining to the users' content request distributions. Moreover, using a sublinear algorithm enables the BBUs to determine the similarity of two users' content request distributions by only scanning a portion of each content request distribution. Compared to traditional stochastic techniques, a sublinear algorithm can control the tradeoff between algorithm processing time or space, and algorithm output quality. Such algorithms can use only a few samples to compute the average content request percentage within the entire content request distributions of all users. %Consequently, sublinear algorithms are promising candidates for cache and clustering problems in CRAN.

Next, we first begin by describing how to use sublinear algorithm for caching. Then, we introduce the entire process using ESNs and sublinear algorithms used to solve (\ref{eq:sum}).

\subsection{Sublinear Algorithm for Clustering and Caching} \label{al:sub}  
In order to cluster the RRHs based on the users' content request distributions and determine which content to cache at the RRHs and BBUs, we first use the prediction of content request distribution and mobility for each user resulting from the output of the ESN schemes to cluster the RRHs and determine which content to cache at RRHs. {The detailed clustering step is specified as follows:
\begin{itemize}
\item The cloud predicts the users' content request distribution and mobility patterns.  
\item Based on the users' content request distribution and locations, the cloud can estimate the users' RRH association. 
\item Based on the users' RRH association, the cloud can determine each RRH's content request distribution and then cluster the RRHs into several groups. For any two RRHs, when the difference of their content request distributions is below $\chi$, the cloud will cluster these two RRHs into the same group. Here, we use the sublinear Algorithm 8 in \cite{Sublinear} to calculate the difference between two content request distributions.  
\end{itemize}} 

Based on the RRHs' clustering, we compute the average of the content request percentage of all users and we use this percentage to determine which content to cache in the cloud. %The massive volume of data is due to the fact that the CRAN will perform a centralized processing of the thousands of users contents transmissions level of the BBUs.     
Based on the prediction of content request distribution and mobility for each user resulting from the output of the ESN schemes, each RRH must determine the contents to cache according to the ranking of the average content request percentage of its associated users, as given by the computed percentages. For example, denote $\boldsymbol{p}_{r,1}$ and $\boldsymbol{p}_{r,2}$ as the prediction of content request distribution for two users that are associated with RRH $r$. The average content request percentage is given as $\boldsymbol{p}_r={{({\boldsymbol{p}_{r,1}} + {\boldsymbol{p}_{r,2}})} \mathord{\left/
 {\vphantom {{({p_1} + {p_2})} 2}} \right.
 \kern-\nulldelimiterspace} 2}$. Based on the ranking of the average content request percentage of the associated users, the RRH selects $C_r$ contents to store in the cache as follows:
\begin{equation}\label{eq:Cr}
\setlength{\abovedisplayskip}{2 pt}
\setlength{\belowdisplayskip}{3 pt}
\mathcal{C}_{r}=\mathop {\arg\max }\limits_{\mathcal{C}_{r}} \sum\limits_{n \in \mathcal{C}_{r} } {{p_{rn}}},
\end{equation}
where $p_{rn}={\sum\nolimits_{i \in {\mathcal{U}_r}} {p_{rin}} E_{k,i}(\theta _{i,n,k}^O)  \mathord{\left/
 {\vphantom {\sum\nolimits_{i \in {\mathcal{U}_r}} {{\boldsymbol{p}_i}}  N_r}} \right.
 \kern-\nulldelimiterspace} N_r}$ is the average weighted percentage of the users that are associated with RRH $r$ requesting content $n$, $\mathcal{U}_r$ is the set of users that are associated with RRH $r$, and $N_r$ is the number of users that are associated with RRH $r$.

To determine the contents that must be cached at cloud, the cloud needs to update the content request distribution of each user to compute the distribution of the requested content that must be transmitted via fronthaul links based on the associated RRH cache. We define the distribution of the requested content that must be transmitted via fronthaul links using the updated content request distribution, $\boldsymbol{p}'_{r,1}=\left[p'_{r11},\dots,p'_{r1N}\right]$. The difference between $\boldsymbol{p}_{r,1}$ and $\boldsymbol{p}'_{r,1}$ is that $\boldsymbol{p}_{r,1}$ contains the probability of the requested content that can be transmitted from the RRH cache. For example, we assume that content $n$ is stored at the cache of RRH $r$, which means that content $n \in \mathcal{C}_{r}$, consequently, $p'_{r1n}=0$. Based on the updated content request distribution, the BBUs can compute the average percentage of each content within the entire content request distributions. For example, let $\boldsymbol{p}'= \sum\nolimits_{\tau  = 1}^{{T }}{{\sum\nolimits_{i = 1}^U {{\boldsymbol{p}'_{\tau,i}E_{k,i}(\theta _{i,k}^A)}} } \mathord{\left/
 {\vphantom {{\sum\nolimits_{i = 1}^R {{p_{\tau,i}}} } {TU}}} \right.
 \kern-\nulldelimiterspace} {TU}}$ be the average of the updated content request probability during $T$, where $\boldsymbol{p}'_{\tau,i}$ is the updated content request distribution of user $i$ during time slot $\tau$. Consequently, the BBUs select $C_c$ contents to store at the cloud cache according to the rank of the average updated content request percentage $\boldsymbol{p}'$ which is:
 \begin{equation}\label{eq:Cc}
 \setlength{\abovedisplayskip}{4 pt}
\setlength{\belowdisplayskip}{4 pt}
\mathcal{C}_{c}=\mathop {\arg \max} \limits_{\mathcal{C}_{c}}  \sum\limits_{n \in \mathcal{C}_{c} } {{p'_{n}}}.
\end{equation} 

However, within a period $T$, the BBUs cannot record the updated content request distributions for all of the users as this will lead to a massive amount of data that is equal to $N\cdot U\cdot T$. The sublinear approach can use only a few updated content request distributions to approximate the actual average updated content request percentage. Moreover, the sublinear approach can control the deviation from the actual average updated content request percentage as well as the approximation error. Since the calculation of the percentage of each content is independent and the method of computing each content is the same, we introduce the percentage calculation of one given content. We define $\epsilon$ as the error that captures the deviation from the actual percentage of each content request. Let $\delta$ be a confidence interval which denotes the probability that the result of sublinear approach exceeds the allowed error interval. To clarify the idea, we present an illustrative example. For instance, assume that the actual result for the percentage of content $n$ is $\alpha=70\%$ with $\epsilon=0.03$ and $\delta=0.05$. This means that using a sublinear algorithm to calculate the percentage of content request of type $n$ can obtain a result whose percentage ranges from $67\%$ to $73\%$ with 95\% probability. Then, the relationship between the number of the updated content request distributions $N_n$ that a sublinear approach needs to calculate the percentage of content $n$, $\epsilon$, and $\delta$ can be given by \cite{Sublinear}: 
\begin{equation}\label{eq:sublinear}
\setlength{\abovedisplayskip}{0 pt}
\setlength{\belowdisplayskip}{3 pt}
N_{n}=- \frac{{\ln \delta }}{{2{\epsilon ^2}}}.
\end{equation}

From (\ref{eq:sublinear}), we can see that a sublinear algorithm can transform a statistical estimation of the expected value into a bound with error deviation $\epsilon$ and confidence interval $\delta$. After setting $\epsilon$ and $\delta$, the sublinear algorithm can just scan $N_n$ updated content request distributions to calculate the average percentage of each content. Based on the average updated content request percentage, the BBUs store the contents that have the high percentages.

\subsection{Proposed Framework based on ESN and Sublinear Approaches}
In this subsection, we formulate the proposed algorithm to solve the problem in (\ref{eq:sum}). First, the BBUs need to run the ESN algorithm to predict the  distribution of content requests and mobility pattern for each user as per Section \ref{section2}, and determine which content to store in RRH cache based on the average content request percentage of the associated users at each time slot. Then, based on the content request distribution of each user, the BBUs cluster the RRHs and sample the updated content request distributions to calculate the percentage of each content based on (\ref{eq:sublinear}). Finally, the BBUs uses the approximated average updated content request percentage to select the contents which have the high percentages to cache at cloud. Based on the above formulations, the algorithm based on ESNs and sublinear algorithms is shown in Algorithm \ref{al2}. Note that, in step 8 of Algorithm 1, a single RRH may belong to more than one cluster since its associated users may have different content request distribution. As an illustrative example, consider a system having two RRHs: an RRH $a$ has two users with content request distributions $\boldsymbol{p}_{a,1}$ and $\boldsymbol{p}_{a,2}$, an RRH $b$ has two users with content request distribution $\boldsymbol{p}_{b,1}$ and $\boldsymbol{p}_{b,3}$, and an RRH $c$ that is serving one user with content request distribution $\boldsymbol{p}_{c,2}$. If $\boldsymbol{p}_{a,1}=\boldsymbol{p}_{b,1}$ and $\boldsymbol{p}_{a,2}=\boldsymbol{p}_{c,2}$, the BBUs will group RRH $a$ and RRH $b$ into one cluster ($\boldsymbol{p}_{a,1}=\boldsymbol{p}_{b,1}$) and RRH $a$ and RRH $c$ into another cluster ($\boldsymbol{p}_{a,2}=\boldsymbol{p}_{c,2}$). In this case, the RRHs that are grouped into one cluster will have the highest probability to request the same contents.

In essence, 
%the proposed sublinear approach allows computation of the percentage of each content request by using only a few samples of content request distributions which reduces the BBUs' scanning time for each content request distribution. Also, 
caching the contents that have the high percentages means that the BBUs will encourage more users to receive the contents from the cache. From (\ref{eq:vBU}), we can see that storing the contents in the RRH cache and cloud cache can reduce the backhaul and fronthaul traffic of each content that is transmitted from the content server and BBUs to the users. Consequently, caching increases the backhaul rate $v_{BU}$ and $v_{FU}$ which will naturally result in a reduction of $\theta$ and an improvement in the effective capacity. We will show next that the proposed caching Algorithm 1 would be an optimal solution to the problem. 
%Moreover, since $\theta _l^H<\theta _l^S$, encouraging more users to receive the contents from the cache can also increase $E_{s}$ based on (\ref{eq:Se}). Therefore, finding the contents which have the high percentages to cache would find a solution to the problem. %Note that the performance of the proposed algorithm can be improved by incorporating a training sequence to update the output weight matrix $\boldsymbol{W}^{out}$. Adjusting $\boldsymbol{W}^{in}$ and the recurrent matrix $\boldsymbol{W}$ appropriately will also improve the accuracy.
\begin{algorithm}[t]\footnotesize
\caption{Algorithm with ESNs and sublinear algorithms }
\label{al2}
\begin{algorithmic} [1] %这个1 表示每一行都显示数字  
\REQUIRE The set of users' contexts, $\boldsymbol{x}_{t}$ and $\boldsymbol{m}_{t}$;\\ 
\vspace{1pt}  
\ENSURE initialize $\boldsymbol{W}_j^{\alpha,in}$, $\boldsymbol{W}_j^{\alpha}$, $\boldsymbol{W}_j^{\alpha,out}$, $\boldsymbol{W}_j^{in}$, $\boldsymbol{W}_j$, $\boldsymbol{W}_j^{out}$, $\boldsymbol{y}_{j}=0$, $\boldsymbol{s}_{j}=0$, $\epsilon$, and $\delta$ \\ %算法的输出：Output  
\FOR {time $T_\tau$}
\STATE update the output weight matrix $\boldsymbol{W}_{T_\tau+1,j}^{out}$ based on (\ref{eq:w2})
\STATE obtain prediction $\boldsymbol{s}_{T_\tau+1,j}$ based (\ref{eq:y2})
\FOR {time $\tau$} 
%\IF{$rand(.) < \epsilon$}
\STATE obtain prediction $\boldsymbol{y}_{\tau+1,j}$ based on (\ref{eq:es})
\STATE update the output weight matrix $\boldsymbol{W}_{\tau+1,j}^{\alpha,out}$ based on (\ref{eq:updatew})
\STATE determine which content to cache in each RRH based on (\ref{eq:Cr})   
\STATE cluster the RRHs 
\ENDFOR 
\STATE  calculate the content percentage for each content based on (\ref{eq:sublinear}) 
 \STATE determine which content to cache in cloud based on (\ref{eq:Cc})
 \ENDFOR 
\end{algorithmic}
\end{algorithm} 
For the purpose of evaluating the performance of the proposed Algorithm 1, we assume that the ESNs can predict the content request distribution and mobility for each user accurately, which means that the BBUs have the entire knowledge of the location and content request distribution for each user. %Moreover, we ignore the loss of transmitting the contents to store in the caches. 
Consequently, we can state the following theorem:

\begin{theorem}\label{theorem2}
\emph{Given the accurate ESNs predictions of the mobility and content request distribution for each user, the proposed Algorithm \ref{al2} will reach an optimal solution to the optimization problem in (\ref{eq:sum})}.
\end{theorem} 
\begin{proof} See Appendix D.
\end{proof}
\vspace{-0.2cm}

\subsection{Complexity and Overhead of the Proposed Approaches} 
In terms of complexity, for each RRH cache replacement action, the cloud needs to implement $U$ ESN algorithms to predict the users' content request distribution. For each cloud caching update, the cloud needs to implement $U$ ESN algorithms to predict the users' mobility patterns. During each time duration for cached content replacement, $T_\tau$, the cached contents stored at an RRH cache will be replaced $\frac{T_\tau}{\tau}$ times. Therefore, the complexity of Algorithm 1 is $O(U \times \frac{T_\tau}{\tau})$. However, it is a learning algorithm which can build a relationship between the users' contexts and behavior. After the ESN-based algorithm builds this relationship, the ESN-based algorithm can directly output the prediction of the users' behavior without any additional training.
 Here, we note that the running time of the approach will decrease once the training process is completed. 

Next, we investigate the computational overhead of Algorithm 1, which is summarized as follows: a) {\em Overhead of users information transmission between users and the content server:} The BBUs will collect all the users' behavior information and the content server will handle the users' content request at each time slot. However, this transmission incurs no notable overhead because, in each time slot, the BBUs need to only input the users' information to the ESN and the cloud has to deal with only one content request for each user. b) {\em Overhead of content transmission for RRH caching update and cloud caching update:} The content servers need to transmit the most popular contents to the RRHs and BBUs. However, the contents stored at RRH cache and cloud cache are all updated during off-peak hours. At such off-peak hours, the fronthaul and backhaul traffic loads will already be low and, thus, having cache updates will not significantly increase the traffic load of the content transmission for caching. %However, we focus on the video show content transmission. Therefore, the time interval between two content requests of each user is long enough to transmit the most popular contents to the RRH cache and cloud cache. 
c) {\em Overhead of the proposed algorithm:} As mentioned earlier, the total complexity of Algorithm 1 is $O(U \times \frac{T_\tau}{\tau})$. Since all the algorithm is implemented at the BBUs which has high-performance processing capacity, the overhead of Algorithm 1 will not be significant. 

\section{Simulation Results} 
For simulations, the content request data that the ESN uses to train and predict content request distribution is obtained from \emph{Youku} of \emph{China network video index}\footnote{\vspace{-0.7cm}The data is available at \url{http://index.youku.com/}.}. The detailed parameters are listed in Table  \uppercase\expandafter{\romannumeral1}. The mobility data is measured from real data generated at the \emph{Beijing University of Posts and Telecommunications}. {Note that the content request data and mobility data sets are independent. To map the data, we record the students' locations during each day and map arbitrarily the students' locations to one user' content request activity from Youku.} The results are compared to three schemes \cite{Content}: a) optimal caching strategy with complete information, b) random caching with clustering, and c) random caching without clustering. All statistical results are averaged over 5000 independent runs. Note that, the benchmark algorithm a) is based on the assumption that the CRAN already knows the entire content request distribution and mobility pattern. Hereinafter, we use the term ``error" to refer to the sum deviation from the estimated distribution of content request to its real distribution. 

\begin{table}\footnotesize
\renewcommand\arraystretch{1}
 \setlength{\abovecaptionskip}{0.9pt}
 \caption{
    \vspace*{-0em} SYSTEM PARAMETERS}\vspace*{-0.3em}
\centering  
\newcommand{\tabincell}[2]{\begin{tabular}{@{}#1@{}}#1\end{tabular}}
\begin{tabular}{|c|c|c|c|}% ±íÊŸž÷ÁÐÔªËØ¶ÔÆë·œÊœ£¬left-l,right-r,center-c
\hline
\textbf{Parameters} & \textbf{Values} & \textbf{Parameters} & \textbf{Values}\\
\hline
$r$ & 1000 m & $P$ & 20 dBm\\
\hline
$R$ & 1000 & $\beta$& 4\\
\hline
$B$ & 1 MHz & $\lambda^\alpha$ & 0.01\\
\hline
$L$ & 10 Mbit & $S$ & 25\\
\hline
$\theta _s^O$& 0.05 & $T$ & 300\\
\hline
 $N_w$ & 1000 & $\sigma ^2$ & -95 dBm\\ 
\hline
 $C_c$,$C_r$ & 6,3 & $D_{\max }$ & 1\\ 
\hline
$K$&7&$N_s$&10\\
\hline
$\delta$ & 0.05&$\epsilon$ & 0.05\\
\hline
$H$& 3 & $\lambda$&0.5\\
\hline
$T_\tau$& 30&$\chi$&0.85\\
\hline
\end{tabular}
 \vspace{-0.3cm}
\end{table}

\begin{figure}[!t]
  \begin{center}
   \vspace{0cm}
    \includegraphics[width=7cm]{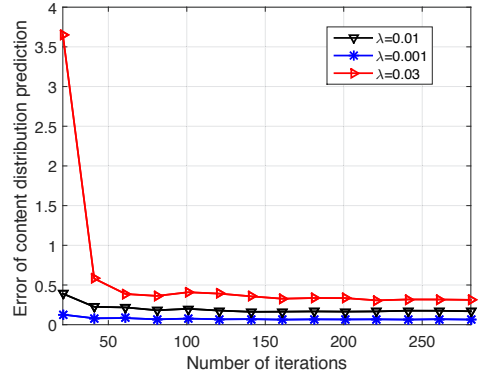}
    \vspace{-0.3cm}
    \caption{\label{Fig3} Error as the number of iteration varies.}
  \end{center}\vspace{-0.7cm}
\end{figure}

\begin{figure*}
\centerline{\subfigure[$N_{tr}=4500$ ]{\includegraphics[width=5.5cm]{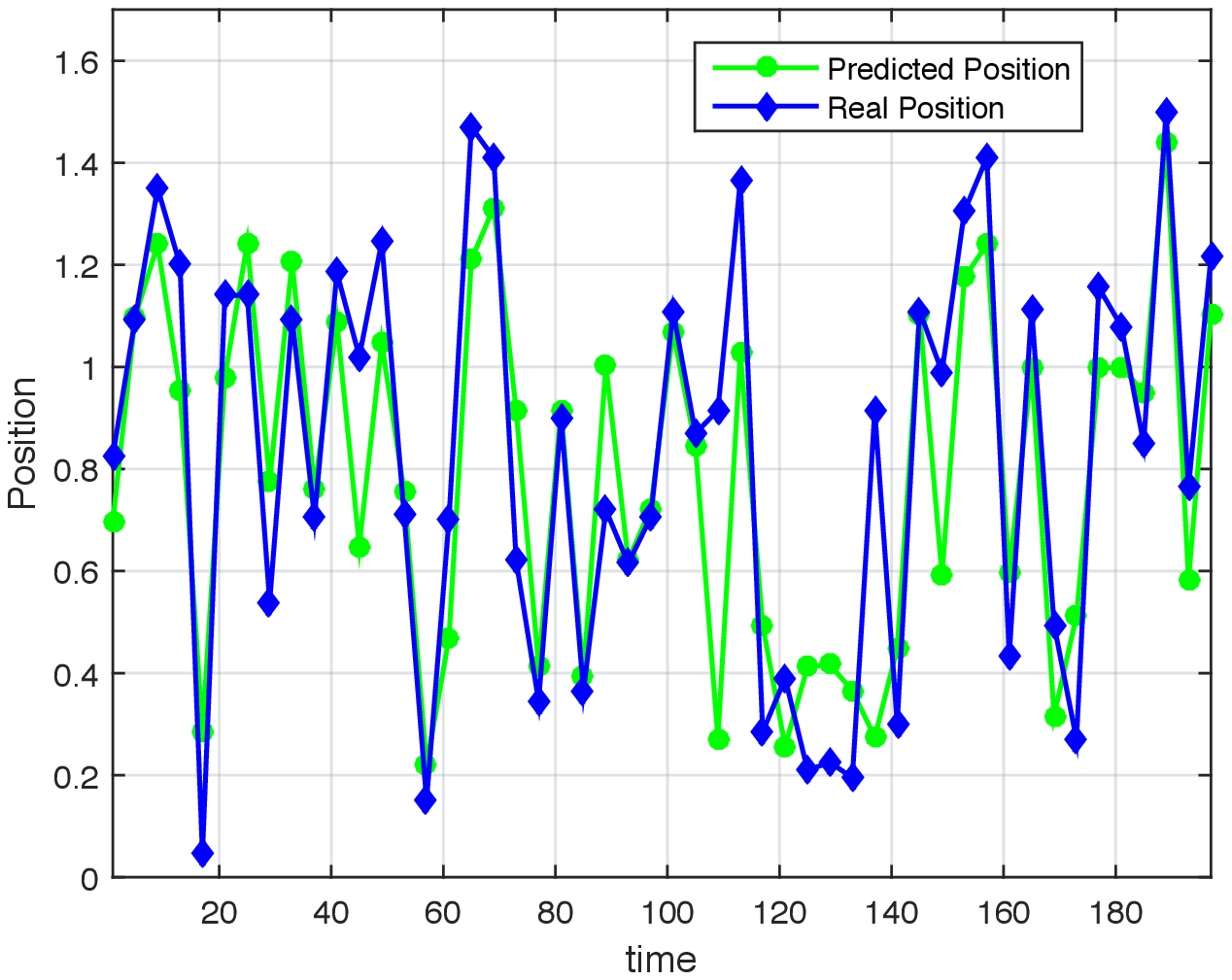}
\label{fig41}}\hspace{-0.65cm}
\subfigure[$N_{tr}=7500$]{\includegraphics[width=5.5cm]{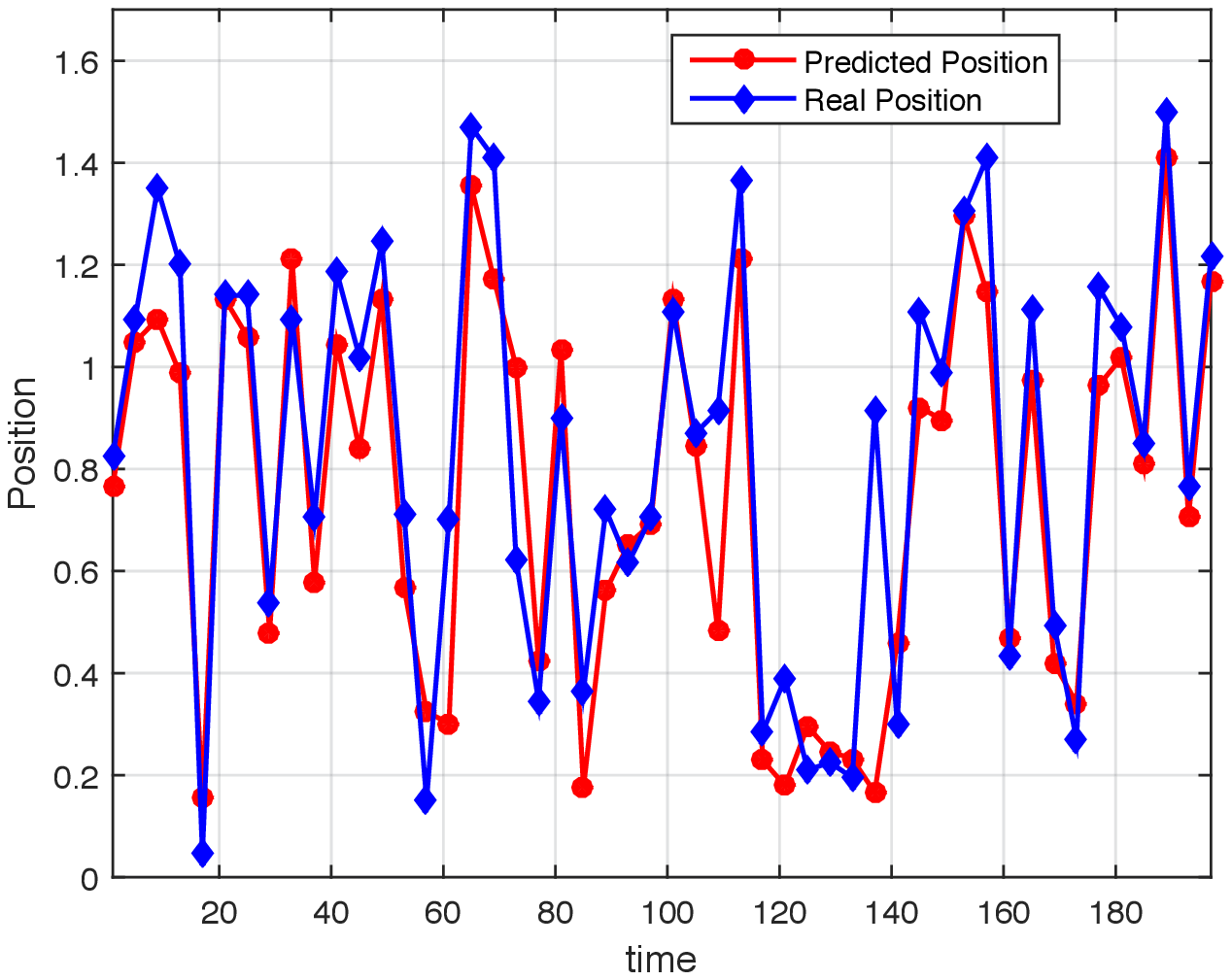}
\label{fig42}}\hspace{-0.65cm}
\subfigure[$N_{tr}=10500$]{\includegraphics[width=5.5cm]{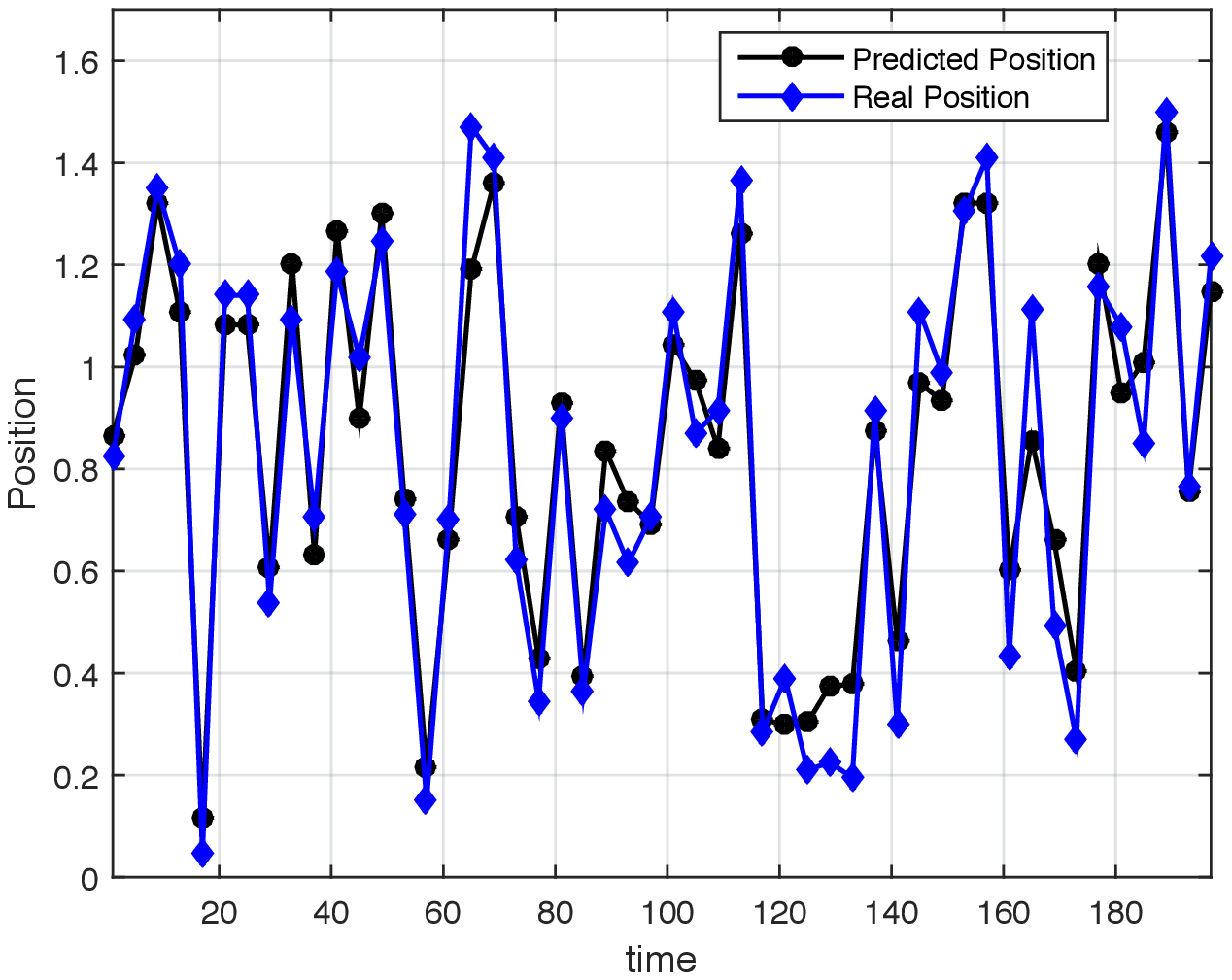}
\label{fig43}}}\vspace{-0.1cm}
\caption{\label{fig4}The ESNs prediction of the users mobility as the training dataset $N_{tr}$ varies.} \vspace{-0.3cm}%\protect\\
\end{figure*}

\begin{figure*}
\centerline{\subfigure[$W=300$ ]{\includegraphics[width=5.5cm]{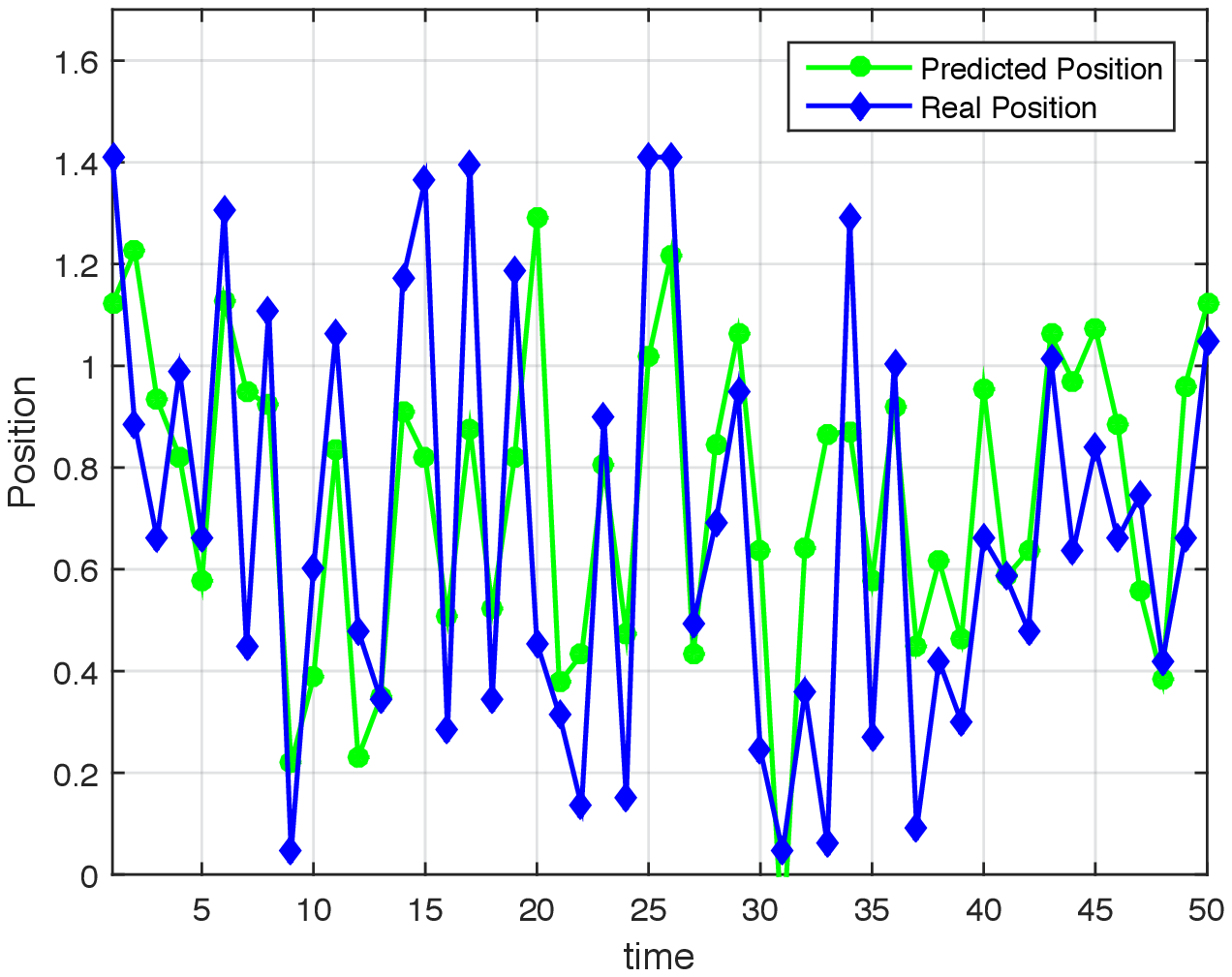}
\label{fig51}}\hspace{-0.65cm}
\subfigure[$W=800$]{\includegraphics[width=5.5cm]{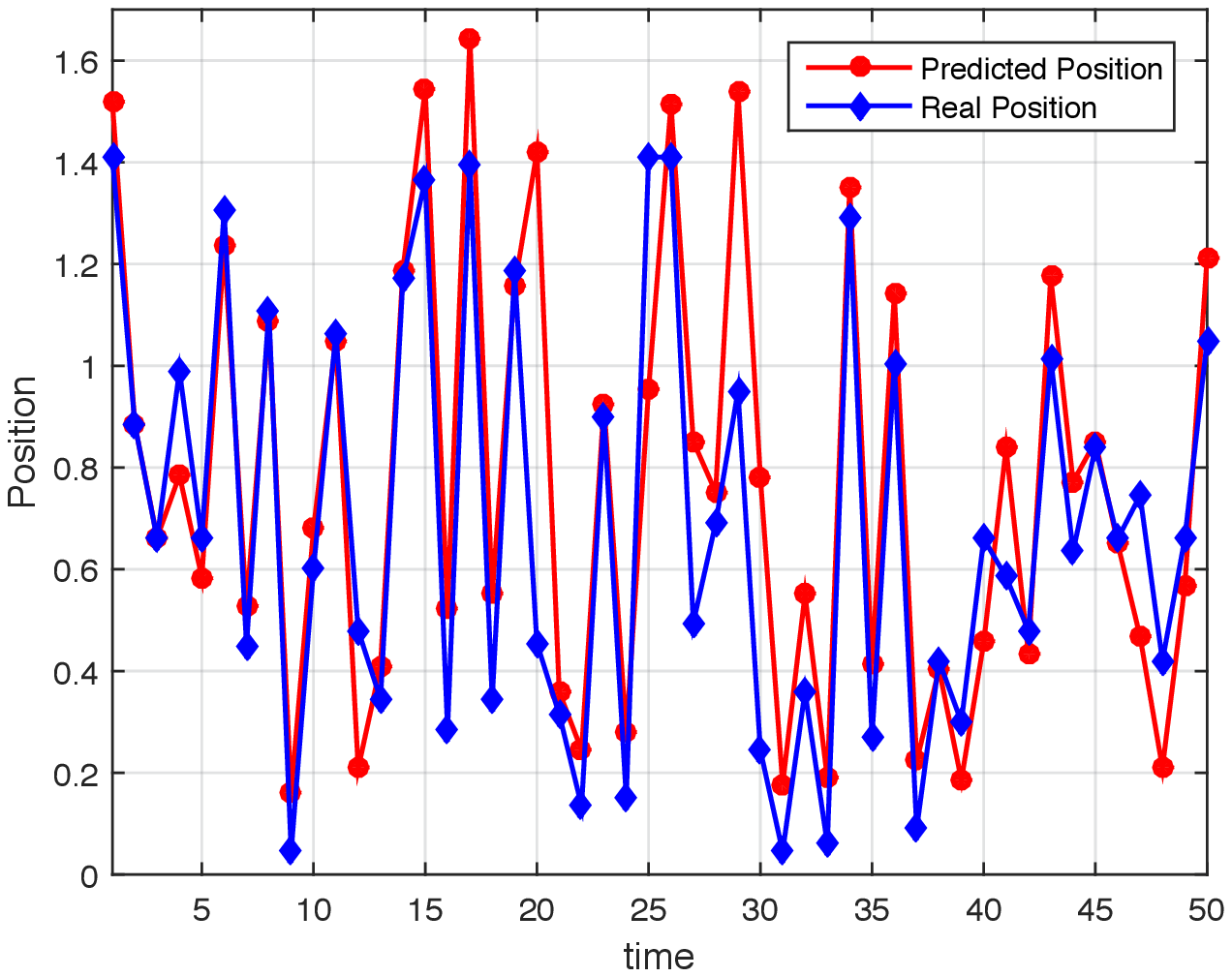}
\label{fig52}}\hspace{-0.65cm}
\subfigure[$W=3000$]{\includegraphics[width=5.5cm]{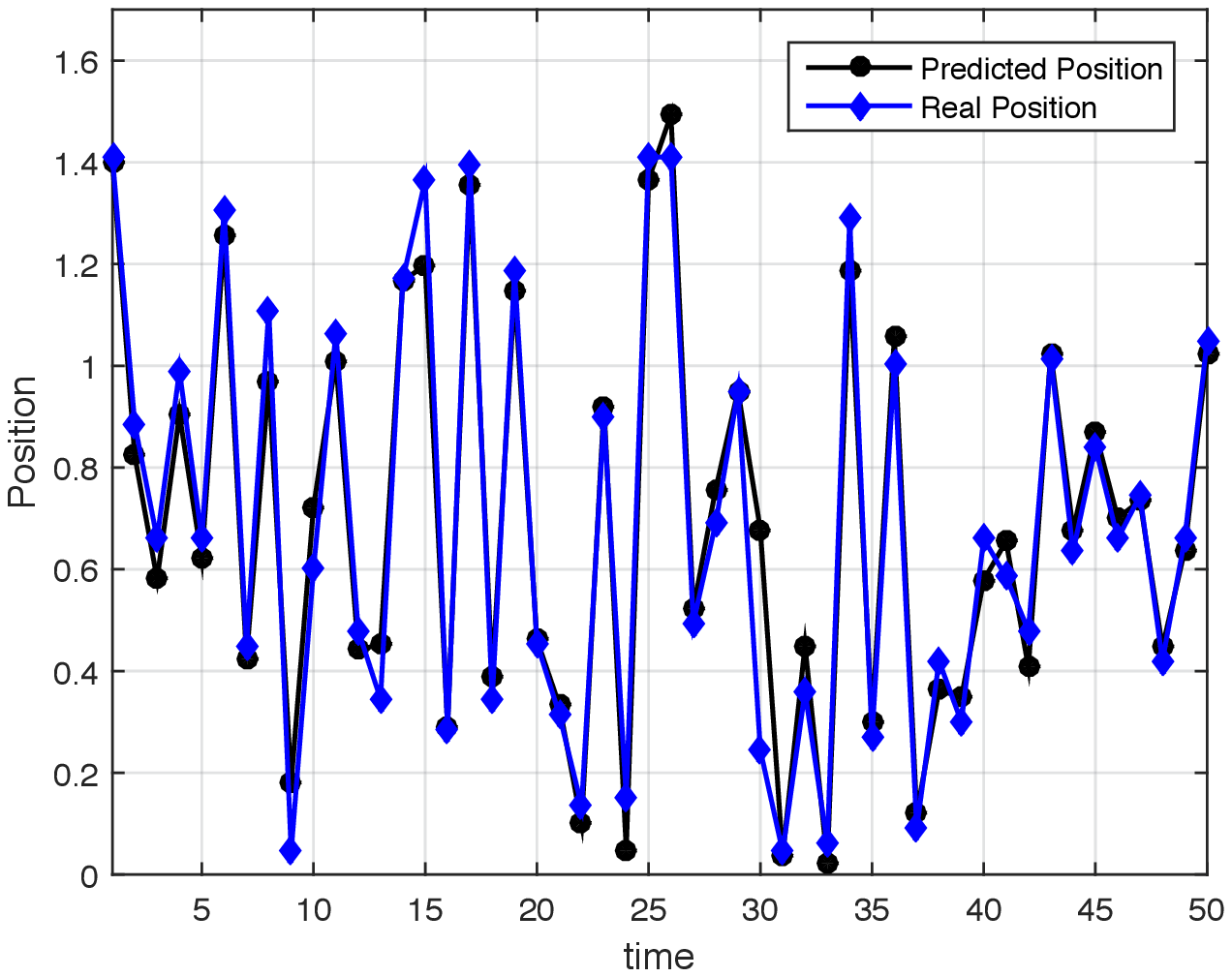}
\label{fig53}}}\vspace{-0.1cm}
\caption{\label{fig5}The ESNs prediction of the users mobility as the ESNs reservoir units varies.}\vspace{-0.3cm}%\protect\\
\end{figure*}

Fig. \ref{Fig3} shows how the error of the ESN-based estimation changes as the number of the iteration varies. In Fig. \ref{Fig3}, we can see that, as the number of iteration increases, the error of the ESN-based estimation decreases. Fig. \ref{Fig3} also shows that the ESN approach needs less than 50 iterations to estimate the content request distribution for each user. This is due to the fact that ESNs need to only train the output weight matrix. Fig. \ref{Fig3} also shows that the learning rates $\lambda^\alpha=0.01,0.001$, and $0.03$ result, respectively, in an error of $0.2\%, 0.1\%$, and $0.43\%$. Clearly, adjusting the learning rates at each iteration can affect the accuracy of the ESNs' prediction.

Figs. \ref{fig4} and \ref{fig5} evaluate the accuracy of using ESN for predicting the users' mobility patterns. First, in Fig. \ref{fig4}, we show how ESN can predict the users' mobility patterns as the size of the training dataset $N_{tr}$ (number of training data to train $\boldsymbol{W}^{out}$) varies. The considered training data is the user's context during a period. In Fig. \ref{fig4}, we can see that, as the size of the training dataset increases, the proposed ESN approach achieves more improvement in terms of the prediction accuracy. Fig. \ref{fig5} shows how ESN can predict users mobility as the number of the ESNs reservoir units $W$ varies. In Fig. \ref{fig5}, we can see that the proposed ESN approach achieves more improvement in terms of the prediction accuracy as the number of the ESNs reservoir units $W$ increases. This is because the number of the ESNs reservoir units $W$ directly affects the ESN memory capacity which directly affects the number of user positions that the ESN algorithm can record. Therefore, we can conclude that the choice of an appropriate size of  the training dataset and an appropriate number of the ESNs reservoir units are two important factors that affect the ESN prediction accuracy of the users' mobility patterns. 

Fig. \ref{learningfigure} shows how the prediction accuracy of a user in a period changes
as the number of the hidden units varies. Here, the hidden units of the ESN represents the size of the reservior units.
From Fig. \ref{learningfigure}, we can see that the prediction of the ESN-based learning algorithm is is more accurate compared to the deep learning algorithm and this accuracy improves as the number of the hidden units increases. In particular, the ESN-based algorithm can yield up to of 14.7\% improvement in terms of the prediction accuracy compared with a deep learning algorithm. This is due to the fact that the ESN-based algorithm can build the relationship between the prediction and the position that the user has visited which is different from the deep learning algorithm that just records the property of each user's locations. Therefore, the ESN-based algorithm can predict the users' mobility patterns more accurately. 
\begin{figure}
  \begin{center}
   \vspace{0cm}
    \includegraphics[width=7cm]{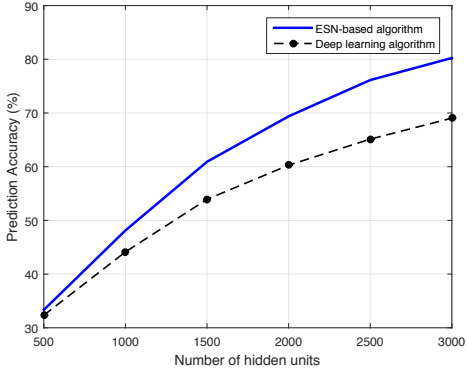}
    \vspace{-0.3cm}
    \caption{\label{learningfigure}Prediction accuracy of mobility patterns as the number of hidden units varies. {Here, we use the deep learning algorithm in \cite{nguyen2012extracting} as a benchmark. The total number of hidden units in deep learning is the same as the number of reservoir units in ESN.}}
  \end{center}\vspace{-0.8cm}
\end{figure}

 %However, by comparing Fig. \ref{fig4a} with Fig. \ref{fig4c}, we can see that the approximation of {\color{blue}ESN} $\alpha$ with $\lambda_\alpha=0.15$ {\color{blue}requires} more than {\color{blue}1500} iterations to approximate the reward function, while, for $\lambda_\alpha=0.08$, it only needs {\color{blue}500} iterations. Clearly, when the learning rate $\lambda_\alpha$ is too large, the update value for the output matrix of {\color{blue}ESN} $\alpha$ is also large, which results in a low speed of convergence. Therefore, we can conclude that {\color{blue}the choice of an appropriate learning rate is an important factor that affects} the convergence speed of {\color{blue}ESN} $\alpha$. 

%This is due to the fact that the size of the training dataset $N_{tr}$ directly affects the training of the ESN output weight matrix. 

\begin{figure}
  \begin{center}
   \vspace{0cm}
    \includegraphics[width=7cm]{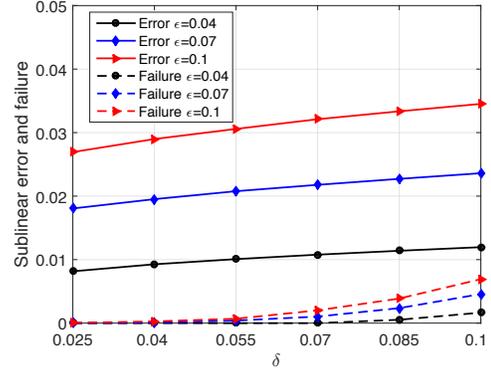}
    \vspace{-0.3cm}
    \caption{\label{Fig6} Error and failure as confidence and allowable error exponents vary.}
  \end{center}\vspace{-0.8cm}
\end{figure}

In Fig. \ref{Fig6}, we show how the failure and error of the content request distribution for each user vary with the confidence exponent $\delta$ and the allowable error exponent $\epsilon$. Here, the error corresponds to the difference between the result of the sublinear algorithm and the actual content request distribution while failure pertains to the probability that the result of our sublinear approach exceeds the allowable error $\epsilon$. From Fig. \ref{Fig6}, we can see that, as $\delta$ and $\epsilon$ increase, the probabilities of failure and error of the content request distribution also increase. This is due to the fact that, as $\delta$ and $\epsilon$ increase, the number of content request distribution samples that the sublinear approach uses to calculate the content percentage decreases. Fig. \ref{Fig6} also shows that even for a fixed $\epsilon$, the error also increase as $\delta$ increases. This is because, as $\delta$ changes, the number of content request distribution samples would also change, which increases the error. 

\begin{figure}
  \begin{center}
   \vspace{0cm}
    \includegraphics[width=7cm]{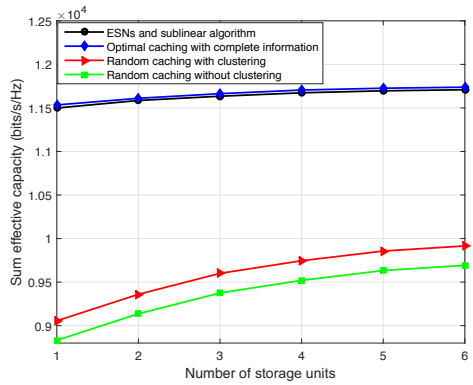}
    \vspace{-0.3cm}
    \caption{\label{Fig8} Sum effective capacity vs. the number of the storage units at cloud cache.}
  \end{center}\vspace{-0.7cm}
\end{figure}

Fig. \ref{Fig8} shows how the sum of the effective capacities of all users in a period changes as the number of the storage units at the cloud cache varies.
In Fig. \ref{Fig8}, we can see that, as the number of the storage units increases, the effective capacities of all considered algorithms increase since having more storages allows offloading more contents from the content server, which, in turn, will increase the effective capacity for each content. From Fig. \ref{Fig8}, we can see that the proposed algorithm can yield up to of $27.8\%$ and $30.7\%$ improvements in terms of the sum effective capacity compared with random caching with clustering and random caching without clustering for the case with one cloud cache storage unit. These gains are due to the fact that the proposed approach can store the contents based on the ranking of the average updated content request percentage of all users as computing by the proposed ESNs and sublinear algorithm.

%\begin{figure}[!t]
%  \begin{center}
%   \vspace{0cm}
%    \includegraphics[width=8cm]{figure9}
%    \vspace{-0.6cm}
%    \caption{\label{Fig9} Sum effective capacity vs. the number of the storage units at each RRH cache.}
%  \end{center}\vspace{-1.1cm}
%\end{figure}
%
%Fig. \ref{Fig9} shows how the sum of the effective capacities of all users in a period changes as the number of the storage units at the RRH cache varies. In Fig. \ref{Fig9}, we can see that, as the number of the storage units increases, the effective capacities of all considered algorithms increase as observed for the cloud-based case in Fig. \ref{Fig8}. In Fig. \ref{Fig9}, we can see that the proposed algorithm can yield up to $18.3\%$ and $22.2\%$ improvements in terms of the sum effective capacity compared, respectively, to random caching with clustering and random caching without clustering for one cache storage unit for each RRH. These gains are due to the fact that the proposed approach can store the contents based on the ranking of the average content request percentage of each RRH as computing by the proposed ESN algorithm. 

\begin{figure}
  \begin{center}
   \vspace{0cm}
    \includegraphics[width=7cm]{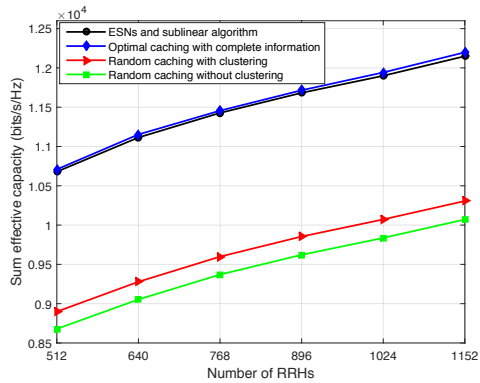}
    \vspace{-0.3cm}
    \caption{\label{Fig10} Sum effective capacity vs. the number of the RRHs.}
  \end{center}\vspace{-0.7cm}
\end{figure}

Fig. \ref{Fig10} shows how the sum of the effective capacities of all users in a period changes as the number of the RRHs varies. In Fig. \ref{Fig10}, we can see that, as the number of the RRHs increases, the effective capacities of all algorithms increase since having more RRHs reduces the distance from the user to its associated RRH. In Fig. \ref{Fig10}, we can also see that the proposed approach can yield up to 21.6\% and 24.4\% of improvements in the effective capacity compared to random caching with clustering and random caching without clustering, respectively, for a network with 512 RRHs. Fig. \ref{Fig10} also shows that the sum effective capacity of the proposed algorithm is only $0.7\%$ below the optimal caching scheme that has a complete knowledge of content request distribution, mobility pattern, and the real content request percentage. Clearly, the proposed algorithm reduces running time of up to $34\%$ and only needs 600 samples of content request to compute the content percentage while only sacrificing $0.7\%$ network performance.

\begin{figure}[!t]
  \begin{center}
   \vspace{0cm}
    \includegraphics[width=7cm]{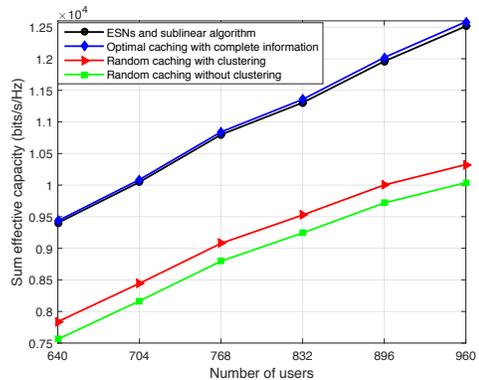}
    \vspace{-0.3cm}
    \caption{\label{Fig11} Sum effective capacity vs. the number of the users.}
  \end{center}\vspace{-0.7cm}
\end{figure}

Fig. \ref{Fig11} shows how the sum of the effective capacities of all users in a period changes as the number of the users varies. In Fig. \ref{Fig11}, we can see that, as the number of the users increases, the effective capacities of all considered algorithms increase as caching can offload more users from the backhaul and fronthaul links as the number of users increases. In Fig. \ref{Fig11}, we can also see that the proposed approach can yield up to 21.4\% and 25\% of improvements in the effective capacity compared, respectively, with random caching with clustering and random caching without clustering for a network with 960 users. This implies that the proposed ESN-based algorithm can effectively use the predictions of the ESNs to determine which content to cache. In Fig. \ref{Fig11}, we can also see that the deviation from the proposed algorithm to the optimal caching increases slightly when the number of users varies. This is due to the fact that the number of content request distributions that the proposed algorithm uses to compute the content percentage is fixed as the total number of content request distributions increases, which will affect the accuracy of the sublinear approximation.

\section{Conclusion}
In this paper, we have proposed a novel caching framework for offloading the backhaul and fronthaul loads in a CRAN system. We have formulated an optimization problem that seeks to maximize the average effective capacities. To solve this problem, we have developed a novel algorithm that combines the machine learning tools of echo state networks with a sublinear caching approach. The proposed algorithm enables the BBUs to predict the content request distribution of each user with limited information on the network state and user context. The proposed algorithm also enables the BBUs to calculate the content request percentage using only a few samples. Simulation results have shown that the proposed approach yields significant performance gains in terms of sum effective capacity compared to conventional approaches. 

\section*{Appendix} \linespread{0.99}
\subsection{Proof of Proposition \ref{pro1}} \label{Ap:a} 
Based on (\ref{eq:thetaD}), the relationship between $\theta_{i,n}^S$ and $\theta_{i,n}^O$ will be:
\begin{equation}\label{eq:1theta}
\setlength{\abovedisplayskip}{5 pt}
\setlength{\belowdisplayskip}{5 pt}
\frac{1}{{\theta _{i,n}^S}} = \frac{1}{{\theta _{i,n}^O}} - \frac{{{{{N_h}L} \mathord{\left/
 {\vphantom {{{N_h}L} v}} \right.
 \kern-\nulldelimiterspace} v}}}{{ - \log \Pr \left( {D > {D_{\max }}} \right)}}.
\end{equation}
Substituting (\ref{eq:PrD}) into (\ref{eq:1theta}), we obtain:
\begin{equation}
\setlength{\abovedisplayskip}{5 pt}
\setlength{\belowdisplayskip}{5 pt}
\frac{1}{{\theta _{i,n}^S}} = \frac{1}{{\theta _{i,n}^O}}\left(1 - \frac{{{N_h}L}}{{{D_{\max }}v}}\right).
\end{equation}
Based on Proposition 5 in \cite{Effectivecapacity}, for the transmission link a), we can take the backhaul transmission rate $v_{BU}$ as the external rate, and consequently, the link hops $N_h$ consists of the link from the BBUs to the RRHs and the link from the RRHs to the users ($N_h=2$). We complete the proof for link a). For link b) and link d), we ignore the delay and QoS losses of the transmission rates from the caches to the BBUs and RRHs, and consequently, the link hops of b) and d) are given as $N_h=1$ and $N_h=2$. The other proofs are the same as above.  

\subsection{Proof of Theorem \ref{theorem1}} \label{ap:b} 
Given an input stream $\boldsymbol{m}(\ldots t)=\ldots m_{t-1}m_t$, where $m_t$ follows the same distribution as $m_{t-W}$, we substitute the input stream $\boldsymbol{m}(\ldots t)$ into (\ref{eq:reservoirstate}), then we obtain the states of the reservoir units at time $t$:
\begin{equation}\small\nonumber
\setlength{\abovedisplayskip}{5 pt}
\setlength{\belowdisplayskip}{5 pt}
\begin{split}
{v_{t,1}} =\; & w_1^{in}{m_t} + w_W^{in}{m_{t - 1}}w +  \cdots  + w_2^{in}{m_{t - (W -1)}}{w^{W - 1}} +  \cdots\\
 &+ w_1^{in}{m_{t - W}}{w^N} +  \cdots  + w_2^{in}{m_{t - (2W - 1)}}{w^{2W - 1}}\\
 &+ w_1^{in}{m_{t - 2W}}{w^{2W}} +  \cdots\\
{v_{t,2}} =\; & w_2^{in}{m_t} + w_1^{in}{m_{t - 1}}w +  \cdots  + w_{3}^{in}{m_{t - (W -1)}}{w^{W - 1}} +  \cdots\\
 &+ w_2^{in}{m_{t - W}}{w^W} +  \cdots  + w_3^{in}{m_{t - (2W - 1)}}{w^{2W - 1}}\\
 &+ w_2^{in}{m_{t - 2N}}{w^{2W}} +  \cdots\\  
\end{split}
\end{equation}

%\begin{equation}\small\nonumber
%\begin{split}
%{v_{t,N}} =\; & w_W^{in}{m_t} + w_{W-1}^{in}{m_{t - 1}}w +  \cdots  + w_{1}^{in}{m_{t - (W -1)}}{w^{W - 1}} +  \cdots\\
% &+ w_W^{in}{m_{t - W}}{w^W} +  \cdots  + w_1^{in}{m_{t - (2W - 1)}}{w^{2W - 1}}\\
% &+ w_W^{in}{m_{t - 2W}}{w^{2W}} +  \cdots 
% \end{split}
%\end{equation} 
Here, we need to note that the ESN having the ability to record the location that the user has visited at time $t-k$ denotes the ESN can output this location at time $t$. Therefore, in order to output $m_{t-k}$ at time $t$, the optimal output matrix $\boldsymbol{W}_j^{out}$ is given as \cite{Short}:
\begin{equation}
\setlength{\abovedisplayskip}{5 pt}
\setlength{\belowdisplayskip}{5 pt}
{\boldsymbol{W}_j^{out}} =\left(\mathbb{E}{\left[ {\boldsymbol{v}_{t,j}{\boldsymbol{v}_{t,j}^{\rm T}}}\right]^{ - 1}}\mathbb{E}\left[ {\boldsymbol{v}_{t,j}m_{t - k} } \right]\right)^{\rm T},
\end{equation}
where $\mathbb{E}{\left[ {\boldsymbol{v}_{t,j}{\boldsymbol{v}_{t,j}^{\rm T}}}\right]}$ is the covariance matrix of $\boldsymbol{W}_{j}^{in}$. Since the input stream is periodic and zero expectation, each element $\mathbb{E}{\left[ {v_{t,i}{v_{t,j}}}\right]}$ of this matrix will be:
\begin{equation}
\setlength{\abovedisplayskip}{5 pt}
\setlength{\belowdisplayskip}{5 pt}
\begin{split}
\mathbb{E}\left[ {{v_{t,i}}{v_{t,j}}}\right]& = w_i^{in}w_j^{in}\sigma _t^2 + w_{i-1(\bmod)W}^{in}w_{j-1(\bmod)W}^{in}\sigma _{t - 1}^2{w^2} \\ &\;\;\;\;\;+  \cdots  + w_i^{in}w_j^{in}\sigma _{t - W}^2{w^{2W}} +  \cdots \\
&= w_i^{in}w_j^{in}\sigma _t^2\sum\limits_{j = 0}^\infty  {{w^{2Wj}}} +  \cdots  \\ &+w_{i-(W-1)}^{in}w_{j-(W-1)}^{in}\sigma _{t - (W-1)}^2\sum\limits_{j = 0}^\infty  {{w^{2Wj + 2(W - 1)}}} \\
&={\boldsymbol{\Omega} _i}\boldsymbol{\Gamma} {\boldsymbol{\Omega}_j^{\rm T}},
\end{split}
\end{equation}
where
\begin{equation}\small\nonumber
\setlength{\abovedisplayskip}{5 pt}
\setlength{\belowdisplayskip}{5 pt}
\boldsymbol{\Gamma}=\left[\!\! {\begin{array}{*{20}{c}}
{\sigma _t^2\sum\limits_{j = 0}^\infty  {{\mathbb{E}\left[w^{2Wj}\right]}} }&0&0\\
0& \ddots &0\\
0&0&{\sigma _{t - (W - 1)}^2\sum\limits_{j = 0}^\infty  {{\mathbb{E}[w^{2Wj + 2(W - 1)}]}} }
\end{array}} \!\right],
\end{equation}
 $\boldsymbol{\Omega}_j$ indicates row $j$ of $\boldsymbol{\Omega}$, $v_{t,j}$ is the element of $\boldsymbol{v}_{t,j}$, and $\sigma _{t -k}^2$ is the variance of $m_{t-k}$. Consequently, $\mathbb{E}{\left[ {\boldsymbol{v}_{t,j}{\boldsymbol{v}_{t,j}^{\rm T}}}\right]}={\boldsymbol{\Omega}}\boldsymbol{\Gamma} {\boldsymbol{\Omega}^{\rm T}}$, $\mathbb{E}\left[ {\boldsymbol{v}_{t,j}m_{t - k} }\right]=\mathbb{E}\left[{w^k}\right]\sigma _{t - k}^2{\boldsymbol{\Omega}_{k+1(\bmod)W}^{\rm T}}$ and ${\boldsymbol{W}^{out}} = \mathbb{E}\left[{w^k}\right]\sigma _{t - k}^2\boldsymbol{\Omega} _{k + 1(\bmod)W}{({\boldsymbol{\Omega}}\boldsymbol{\Gamma} \boldsymbol{\Omega}^{\rm T} )^{ - 1}}$. Based on these formulations and (\ref{eq:y2}), the ESN output at time $t$ will be ${s_{t,j}} = \boldsymbol{W}^{out}{\boldsymbol{v}_{t,j}}=\mathbb{E}\left[{w^k}\right]\sigma _{t - k}^2\boldsymbol{\Omega} _{k + 1(\bmod)W}{\left({\boldsymbol{\Omega}}\boldsymbol{\Gamma} \boldsymbol{\Omega}^{\rm T} \right)^{ - 1}}{\boldsymbol{v}_{t,j}}$.
Consequently, the covariance of ESN output $s_{t,j}$ with the actual input $m_{t-k,j}$ is given as:
\begin{equation}\small\nonumber
\setlength{\abovedisplayskip}{5 pt}
\setlength{\belowdisplayskip}{5 pt}
\begin{split}
&\textrm{Cov}\left({s_{t,j}},{m_{t - k,j}}\right) \\
&=\mathbb{E}\left[{w^k}\right]\sigma _{t - k}^2\boldsymbol{\Omega} _{k + 1(\bmod)W}{\left({\boldsymbol{\Omega}}\boldsymbol{\Gamma} \!\!\boldsymbol{\Omega}^{\rm T} \right)^{ - 1}}\mathbb{E}\left[{\boldsymbol{v}_{t,j}},{m_{t - k}}\right],\\
&=\mathbb{E}\left[{w^k}\right]^2\!\!\sigma _{t - k}^4\!\left(\!\boldsymbol{\Omega} _{k + 1(\!\bmod\!)W}{\left(\!\boldsymbol{\Omega}^{\rm T}\!\right)^{\!-1}}\right)\!\boldsymbol{\Gamma}^{-1}\!\!\left(\boldsymbol{\Omega}^{-1}\boldsymbol{\Omega} _{k + 1(\!\bmod\!)W}^{\rm T}\!\right),\\
&\mathop  = \limits^{(a)} \mathbb{E}\left[{w^k}\right]^2\sigma _{t - k}^2{\left(\sum\limits_{j = 0}^\infty  {\mathbb{E}\left[{w^{2Wj + 2k(\bmod)W}}\right]} \right)^{ - 1}},
\end{split}
\end{equation}
where $(a)$ follows from the fact that ${\boldsymbol{\Omega} _{k + 1(\bmod )W}} = \boldsymbol{e}_{k + 1}^{\rm T}\boldsymbol{\Omega}^{\rm T}$ and ${\boldsymbol{e}_{k + 1}} = (0, \ldots ,{1_{k + 1}},0 \ldots 0)^{\rm T} \in {\mathbb{R}^{W}}$. Therefore, the memory capacity of this ESN is given as \cite{Minimum}:
\begin{equation}\small\nonumber
\setlength{\abovedisplayskip}{5 pt}
\setlength{\belowdisplayskip}{5 pt}
\begin{split}
M &\!= \sum\limits_{k =0}^\infty  {\mathbb{E}\!{{\left[{w^k}\right]}^2}{{\!\left(\sum\limits_{j = 0}^\infty  {\mathbb{E}\!\left[{w^{2Wj + 2k(\!\bmod\!)W}}\right]}\! \right)^{\!\!-1}}}}\!\!\!\!\!\!-\!{{\left(\sum\limits_{j = 0}^\infty  {\mathbb{E}\!\left[{w^{2Wj}}\right]} \!\right)^{\!\!-1}}}\!\!\!\!,\\
&\!\!\!\!\!\!\!=\sum\limits_{k = 0}^{W - 1} \!{\mathbb{E}{{\left[{w^k}\right]}^2}{{\left(\sum\limits_{j = 0}^\infty \! {\mathbb{E}\left[{w^{2Wj + 2k}}\right]} \right)^{ - 1}}}}\!\!\!\!\\
&\!\!\!\!\!\!\!+\!\!\!\!\!\sum\limits_{k = W}^{2W - 1} \!\!{\mathbb{E}{{\left[{w^k}\right]}^2}{{\!\left(\sum\limits_{j = 0}^\infty{\mathbb{E}\!\left[{w^{2Wj + 2k(\bmod)W}}\right]} \!\right)^{ \!\!\!- 1}}}} \!\!\!\!\!+ \! \cdots\!-\!{{\left(\sum\limits_{j = 0}^\infty \! {\mathbb{E}\!\left[{w^{2Wj}}\right]}\!\!\right )^{\!\!-1}}}\!\!\!\!\!,\\
&\!\!\!\!\!\!\!=\!\!\sum\limits_{k = 0}^{W - 1} {{{\!\!\left(\sum\limits_{j = 0}^\infty  {\mathbb{E}\!\!\left[{w^{2Wj + 2k}}\!\right]}\!\! \right)^{ \!\!\!- 1}}}\!\!\sum\limits_{j = 0}^\infty  {\mathbb{E}{{\left[{w^{Wj + k}}\right]}^2}} }\! \!-\!{{\left(\sum\limits_{j = 0}^\infty {\mathbb{E}\left[{w^{2Wj}}\right]} \right)^{\!\!-1}}}\!\!.
\end{split}
\end{equation}
This completes the proof.

 \subsection{Proof of Proposition \ref{pro3}} \label{Ap:c} 
 For \romannumeral1), we first use the distribution that $P\left(w = a \right) = 0.5$ and $P\left(w = -a\right) = 0.5$ to formulate the memory capacity, where $a\in (0,1)$. Then, we discuss the upper bound. Based on the distribution property of $w$, we can obtain that $\mathbb{E}\left[w^{2W}\right]=a^{2W}$ and $\mathbb{E}\left[w^{2W+1}\right]=0$. The memory capacity is given as:
\begin{equation}\small\label{eq:M}
\setlength{\abovedisplayskip}{5 pt}
\setlength{\belowdisplayskip}{5 pt}
\begin{split}
M &\!=\! \sum\limits_{k = 0}^{W - 1} {{{\!\!\left(\sum\limits_{j = 0}^\infty  {\mathbb{E}\!\left[{w^{2Wj + 2k}}\right]}\! \right)^{\!\!\! - 1}}}\!\!\!\sum\limits_{j = 0}^\infty  {\mathbb{E}{{\left[{w^{Wj + k}}\right]}^2}} }\!\!\!-\!{{\left(\sum\limits_{j = 0}^\infty {\mathbb{E}\!\left[{w^{2Wj}}\right]} \!\!\right)^{\!\!\!-1}}}\!\!,\\
&\!\!\!\!\!\!\!\!=\!\!\!\sum\limits_{k = 0}^{W - 1} {{{\!\!\left(\sum\limits_{j = 0}^\infty  {{a^{2Wj + 2k}}} \!\right)^{\!\!\! - 1}}}\!\!\sum\limits_{j = 0}^\infty  {{{{a^{2Wj + 2k}}}}} }\! \!-\!{{\left(\sum\limits_{j = 0}^\infty {{a^{2Wj}}} \!\right)^{\!\!\!-1}}}\!\!\!, \left(\!\text{$k$ is an even}\right),\\
&\!\!\!\!\!\!\!\!=\!\sum\limits_{k = 0}^{\left\lfloor {\frac{W}{2}} \right\rfloor } 1  - \left(1 - {a^{2W}}\right)={\left\lfloor {\frac{W}{2}} \right\rfloor }+a^{2W} <{\left\lfloor {\frac{W}{2}} \right\rfloor }+1.
\end{split}
\end{equation}
From (\ref{eq:M}), we can also see that the memory capacity $M$ increases as both the moment $\mathbb{E}\left[w^k\right]$ and $a$ increase, $k \in \mathbb{Z}^{+}$. This completes the proof of \romannumeral1).  For case \romannumeral2), we can use a similar method to derive the memory capacity exploiting distribution that $P\left(w = a\right) = 1$ and consequently, $\mathbb{E}\left[w^{k}\right]=a^k$, this yielding $M=W-1+a^{2W}$. Since $a \in (0,1)$, $M<W$ which is also correspondent to the existing work \cite{Short}.

\subsection{Proof of Theorem \ref{theorem2}} \label{ap:d} 
The problem based on (\ref{eq:sum}) for each time slot can be rewritten as:
\begin{equation}\label{proofEs}
\setlength{\abovedisplayskip}{5 pt}
\setlength{\belowdisplayskip}{5 pt}
{\bar E} =\frac{1}{{{T }}}\sum\limits_{k = 1}^{{T }} \sum\limits_{i \in \mathcal{U}} { {E_{k,i}\left(\theta _{{i,n_{ik},k}}^j\right)} },
\end{equation}
where $j \in \left\{ {O,A,S,G} \right\}$.
Denote $\boldsymbol{p}_{k,i}=\left[p_{ki1},p_{ki2},\dots,p_{kiN}\right]$ as the content request distribution of user $i$ at time slot $k$, the average effective capacities of the users is given by:
\begin{equation}\small\label{proofEsn}
\begin{split}
\setlength{\abovedisplayskip}{5 pt}
\setlength{\belowdisplayskip}{5 pt}
{\widetilde E_{k}} &=\!\! \sum\limits_{i \in \mathcal{U}} \!\!\left(\sum\limits_{n_{ik} \in {\mathcal{C}_{i}}} \!{{\!p_{kin_{ik}}}} E_{k,i}\!\left(\theta _{i,n_{ik},k}^O\right) \!\!+ \!\! \!\! \!\!\sum\limits_{n_{ik} \in {{{\mathcal{C}_c}} \mathord{\left/
 {\vphantom {{{C_c}} {{C_k}}}} \right.
 \kern-\nulldelimiterspace} {{\mathcal{C}_i}}}} \!\!\!\!\!{{\!p_{kin_{ik}}}} E_{k,i}\!\left(\theta _{i,n_{ik},k}^A\right)\right) \!\!
 \\&+\sum\limits_{i \in \mathcal{U}} \!\!\left(\! \sum\limits_{n_{ik} \in \mathcal{N}'}\!\!{{p_{kin_{ik}}}} E_{k,i}\!\left(\theta _{i,n_{ik},k}^S\right) \!\!+\!\!\!\! \!\!\sum\limits_{n_{ik} \in {{\mathcal{C}'_i}}} \!{{p_{kin_{ik}}}} E_{k,i}\!\left(\theta _{i,n_{ik},k}^G\right)\!\right) ,
 \end{split}
\end{equation}
where $\mathcal{C}_i$ is the set of RRH cache that is associated with user $i$, $\mathcal{N}'$ and $\mathcal{C}'_i$ represent, respectively, the contents that the BBUs arrange to transmit from the content server and remote RRHs cache. Since the transmission from the content server and the remote RRH cache can be scheduled by the BBUs based on Proposition \ref{pro1}, we only need to focus on the transmissions from the cloud cache and RRHs cache to the users which results in the average effective capacities of the users at time slot $k$ as follows: 
\begin{equation}\small\label{proofEsnA}
\setlength{\abovedisplayskip}{5 pt}
\setlength{\belowdisplayskip}{5 pt}
\begin{split}
{\widetilde E_{k}} &=\sum\limits_{i \in \mathcal{U}}    \sum\limits_{n_{ik} \in {\mathcal{C}_{i}}} {{p_{kin_{ik}}}} E_{k,i}\left(\theta _{i,n_{ik},k}^O\right) \\
&\;\;\;\;\;\;\;\;\;\;\;\;\;\;\;\;\;\;\;\;\;\;\;\;\;\;\;\;\;\;+\sum\limits_{i \in \mathcal{U}} \sum\limits_{n_{ik} \in {{{\mathcal{C}_c}} \mathord{\left/
 {\vphantom {{{C_c}} {{C_k}}}} \right.
 \kern-\nulldelimiterspace} {{\mathcal{C}_i}}}} {{p_{kin_{ik}}}} E_{k,i}\left(\theta _{i,n_{ik},k}^A\right)+F,\\
&= { \sum\limits_{r \in \mathcal{	R}}\sum\limits_{i \in \mathcal{U}_r}{\sum\limits_{n_{ik} \in {\mathcal{C}_{r}}} {{p_{kin_{ik}}}} E_{k,i}\left(\theta _{i,n_{ik},k}^O\right)} }  \\
&\;\;\;\;\;\;\;\;\;\;\;\;\;\;\;\;\;\;\;\;\;\;\;\;\;\;\;\;\;\;+\sum\limits_{i \in \mathcal{U}}  \sum\limits_{n_{ik} \in {\mathcal{C}_c}} {{{p_{kin_{ik}}}E_{k,i}\left(\theta _{i,n_{ik},k}^A\right)} } +F,
\end{split}
\end{equation} 
where 
\begin{equation}\small\nonumber
\!F\!=\!\sum\limits_{i \in \mathcal{U}}\!\sum\limits_{n_{ik} \in \mathcal{N}'}\!\! {{\!\!p_{kin_{ik}}}} E_{k,i}(\theta _{i,n_{ik},k}^S)+\! \!\sum\limits_{i \in \mathcal{U}}\sum\limits_{n_{ik} \in {{\mathcal{C}'_i}}} \!\!{{\!p_{kin_{ik}}}} E_{k,i}(\theta _{i,n_{ik},k}^G).
\end{equation} 
 Since $E_{k,i}(\theta _{i,n_{ik},k}^O)$ depends only on $\theta _{i,n_{ik},k}^O$, we can consider it as a constant during time slot $k$ and consequently, we only need to optimize ${\sum\limits_{i \in \mathcal{U}_r}\! {\sum\limits_{n_{ik} \in {\mathcal{C}_{r}}} \!\!\!{{p_{kin_{ik}}}} E_{k,i}(\theta _{i,n_{ik},k}^O)} }$ for each RRH. Therefore, we can select the content that has the maximal value of ${\sum\limits_{i \in \mathcal{U}_r}\! {{{p_{kin_{ik}}}} E_{k,i}(\theta _{i,n_{ik},k}^O)} }$, which corresponds to the proposed RRH caching method in Section \ref{al:sub}. 

Since the contents that are stored in the cloud cache are updated during a period $T$, the optimization of the cloud cache based on (\ref{proofEs}) and (\ref{proofEsnA}) is given as:
\begin{equation}
\setlength{\abovedisplayskip}{4 pt}
\setlength{\belowdisplayskip}{4 pt}
E_{c}=\max\frac{1}{{{T }}}\sum\limits_{k = 1}^{{T }}\sum\limits_{i \in \mathcal{U}}\sum\limits_{n_{ik} \in {{{\mathcal{C}_c}} \mathord{\left/
 {\vphantom {{{C_c}} {{C_k}}}} \right.
 \kern-\nulldelimiterspace} {{\mathcal{C}_i}}}} {{{p_{kin_{ik}}}E_{k,i}(\theta _{i,n_{ik},k}^A)} }.
\end{equation}
Here, the average of the effective capacity is over different contents transmission.
 After obtain the updated content request distribution of each user, we can use the same method to prove that the proposed algorithm can reach to the optimal performance.

\bibliographystyle{IEEEbib}
\bibliography{references}

\end{document}